\documentclass[10pt,letterpaper]{article}
\usepackage[margin=1in]{geometry}
\usepackage[utf8]{inputenc}
\usepackage[T1]{fontenc}
\usepackage[english]{babel}

\usepackage{url,hyperref,lineno,microtype,subcaption}
\usepackage[onehalfspacing]{setspace}
\usepackage{xspace}

\usepackage{graphicx}
\usepackage{graphics}
\usepackage[ruled,vlined,shortend, linesnumbered]{algorithm2e} 
\usepackage[]{algorithmic} 
\usepackage{booktabs}
\usepackage{color}
\usepackage{amsthm,amssymb}
\usepackage{amsmath}
\usepackage[]{caption}
\usepackage{multirow}
\usepackage{mathrsfs}
\usepackage{amsbsy}
\usepackage[mathscr]{eucal} 
\usepackage[flushleft]{threeparttable}
\usepackage{comment}
\usepackage{pifont}
\usepackage{multicol}
\usepackage{soul}
\usepackage{alphalph}
\usepackage{paralist}

\newtheorem{problem}{Problem}
\newtheorem{definition}{Definition}
\newtheorem{lemma}{Lemma}
\newtheorem{theorem}{Theorem}
\newtheorem{example}{Example}

\let\oldnl\nl
\newcommand{\nonl}{\renewcommand{\nl}{\let\nl\oldnl}}

\newcommand{\hide}[1]{}

\newcommand{\tensor}[1]{\boldsymbol{\mathscr{#1}}}   
\newcommand{\mat}[1]{\mathbf{#1}}

\newcommand{\TR}{\tensor{R}}
\newcommand{\TB}{\tensor{B}}

\newcommand{\MA}{\mat{A}}

\newcommand{\bit}{\begin{compactitem}}
	\newcommand{\eit}{\end{compactitem}}
\newcommand{\ben}{\begin{compactenum}}
	\newcommand{\een}{\end{compactenum}}

\newcommand{\method}{\textsc{D-Cube}\xspace}

\newcommand{\mzoom}{\textsc{M-Zoom}\xspace}
\newcommand{\mbiz}{\textsc{M-Biz}\xspace}
\newcommand{\cross}{\textsc{CrossSpot}\xspace}

\newcommand{\fraudar}{\textsc{Fraudar}\xspace}

\newcommand{\densealert}{\textsc{DenseStream}\xspace}
\newcommand{\densestream}{\textsc{DenseAlert}\xspace}

\newcommand{\mass}[1]{M_{#1}}

\newcommand{\density}{\rho}
\newcommand{\densityarinoarg}{\rho_{ari}}
\newcommand{\densitygeonoarg}{\rho_{geo}}
\newcommand{\densitysuspnoarg}{\rho_{susp}}
\newcommand{\densitysurpnoarg}{\rho_{es(\alpha)}}
\newcommand{\densitysurpalpha}[1]{\rho_{es(#1)}}
\newcommand{\densityno}[2]{\rho(#1,#2)}
\newcommand{\densityari}[2]{\rho_{ari}(#1,#2)}
\newcommand{\densitygeo}[2]{\rho_{geo}(#1,#2)}
\newcommand{\densitysusp}[2]{\rho_{susp}(#1,#2)}
\newcommand{\densitysurp}[2]{\rho_{es(\alpha)}(#1,#2)}

\newcommand{\cmark}{\ding{51}}%
\newcommand{\mapreduce}{\textsc{MapReduce}\xspace}
\newcommand{\hadoop}{\textsc{Hadoop}\xspace}

\newcommand{\bcomment}[1]{{\scriptsize \textcolor{black}{\hfill $\vartriangleright$ {\it #1}}}}
\newcommand{\change}{\textcolor{black}}

\begin{document}
	\title{Detecting Group Anomalies in Tera-Scale Multi-Aspect Data \\ via Dense-Subtensor Mining\footnote{The content of the manuscript has been presented in part at the 10th ACM International Conference on Web Search and Data Mining, \cite{shin2017dcube}. In this extended version, we refined  \method with a new parameter $\theta$, and we proved that the time complexity of \method is significantly improved with the refinement (Lemma~\ref{lemma:epsilon} and Theorem~\ref{thm:time:worst}). 	We also proved that, for $N$-way tensors, \method gives an $\theta N$-approximation guarantee for Problem~\ref{defn:problem} (Theorem~\ref{thm:accuracy:guarantee}).
			Additionally, we considered an extra density measure (Definition~\ref{defn:density:susp}) and \change{an extra competitor (i.e., \mbiz);} and we applied \method to three more real-world datasets (i.e., KoWiki, EnWiki, and SWM) and successfully detected edit wars, bot activities, and spam reviews (Tables~\ref{tab:spam}, \ref{tab:blocks:summary}, and \ref{tab:bot}).
			Lastly, we conducted experiments showing the effects of parameters $\theta$ and $\alpha$ on the speed and accuracy of \method \change{in dense-subtensor detection} (Figures~\ref{fig:tradeoff:epsilon} and \ref{fig:tradeoff:alpha}).
			Most of this work was also included in the PhD thesis of Kijung Shin.}}
	\author{Kijung Shin\,$^{1}$, Bryan Hooi\,$^{2}$, Jisu Kim\,$^{3}$, and Christos Faloutsos\,$^{4}$}
	\date{\normalsize $^{1}$ Graduate School of AI and School of Electrical Engineering, KAIST, Daejeon, South Korea, \\ 
		$^{2}$ School of Computing and Institute of Data Science, National University of Singapore, Singapore, \\
		$^{3}$ DataShape, Inria Saclay, Palaiseau, France, \\
		$^{4}$ School of Computer Science, Carnegie Mellon University, Pittsburgh, PA, USA, \\
		kijungs@kaist.ac.kr, bhooi@comp.nus.edu.sg, jisu.kim@inria.fr, christos@cs.cmu.edu
	}
	
	\maketitle
	
	\begin{abstract}
		How can we detect fraudulent lockstep behavior in large-scale multi-aspect data (i.e., tensors)?
		Can we detect it when data are too large to fit in memory or even on a disk?
		Past studies have shown that dense subtensors in real-world tensors (e.g., social media, Wikipedia, TCP dumps, etc.) signal anomalous or fraudulent behavior such as retweet boosting, bot activities, and network attacks.
		Thus, various approaches, including tensor decomposition and search, have been proposed for detecting dense subtensors rapidly and accurately.
		However, existing methods suffer from low accuracy, or they assume that tensors are small enough to fit in main memory, which is unrealistic in many real-world applications such as social media and web.
		
		To overcome these limitations, we propose \method, a disk-based dense-subtensor detection method, which also can run in a distributed manner across multiple machines.
		Compared to state-of-the-art methods, 
		\method is (1) {\bf Memory Efficient}: requires up to {\it 1,561$\times$ less memory} and handles {\it 1,000$\times$ larger} data ({\it 2.6TB}), (2) {\bf Fast}: up to {\it 7$\times$ faster} due to its near-linear scalability, (3) {\bf Provably Accurate}: gives a guarantee on the densities of the detected subtensors, and (4) {\bf Effective}: spotted network attacks from TCP dumps and synchronized behavior in rating data most accurately.
		
		\vspace{1mm}
		
	\end{abstract}

\section{Introduction}
\label{sec:intro}
\begin{figure*}[t]
	\centering
	\centering
	\includegraphics[width=0.8\linewidth]{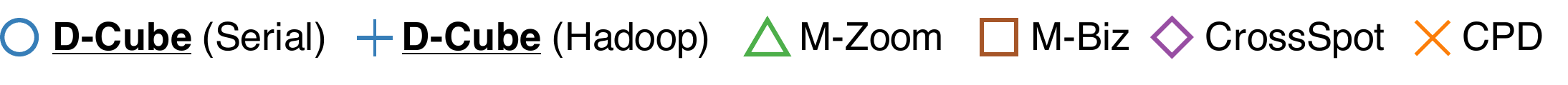} \\
	\begin{tabular}{ccl}
		\begin{minipage}{.2\textwidth}
			\center
			\includegraphics[width=\linewidth]{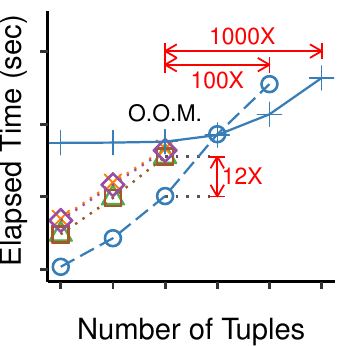}
		\end{minipage} 
		& \begin{minipage}{.2\textwidth}
			\center
			\includegraphics[width=\linewidth]{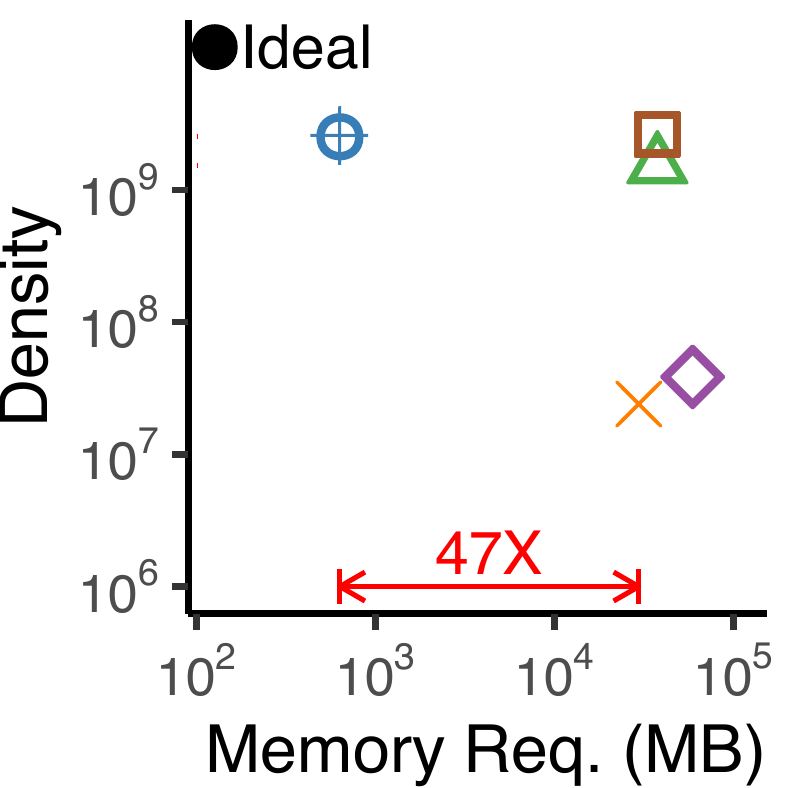}
		\end{minipage}
		& \begin{minipage}{.5\textwidth}
			\scalebox{0.79}{
				\begin{tabular}{cc|cccc}
					\toprule
					Dataset & Order & Volume & Mass & Attack Ratio & Attack Type \\
					\midrule
					\multirow{5}{*}{DARPA}& 1 & 738 & 1.52M  & {\bf 100\%} & Neptune\\
					& 2 & 522 & 614K & {\bf 100\%} & Neptune\\
					& 3 & 402 & 113K & {\bf 100\%} & Smurf\\
					& 4 & 1 & 10.8K & {\bf 100\%} & Satan\\
					& 5 & 156K & 560K & 30.4\% & SNMP\\
					\midrule
					\multirow{5}{*}{AirForce} & 1 & 1 & 1.93M  & {\bf 100\%} & Smurf\\
					& 2 & 8 & 2.53M & {\bf 100\%} &  Smurf\\
					& 3 & 6,160 & 897K & {\bf 100\%} & Neptune\\
					& 4 & 63.5K& 1.02M & 94.7\% & Neptune\\
					& 5 & 930K & 1.00M & 94.7\% & Neptune\\
					\bottomrule
				\end{tabular}
			}
		\end{minipage} \vspace{1mm} \\
		A. Data scalability & B. Accuracy and & \qquad C. Network intrusion detection using \method \\
		& memory requirements & (Top-$5$ subtensors detected by \method in TCP dumps) \\
	\end{tabular}
	\caption{\label{fig:crown}
		\change{{\bf Strengths of \method.} `O.O.M' stands for `out of memory'.
			A. {\bf Fast \& Scalable}: \method was {\it 12$\times$ faster} and successfully handled 
			{\it 1,000$\times$ larger data (2.6TB)} than its best competitors.
			B. {\bf Efficient \& Accurate}: \method required {\it 47$\times$ less memory}
			and found subtensors as dense as those found by its best competitors from English Wikipedia revision history.
			C. {\bf Effective}: \method accurately spotted network attacks from TCP dumps.
			See Section~\ref{sec:experiments} for the detailed experimental settings.}
	}
\end{figure*}

Given a tensor that is too large to fit in memory, how can we detect dense subtensors? Especially, can we spot dense subtensors without sacrificing speed and accuracy provided by in-memory algorithms?

A common application of this problem is review fraud detection, where we aim to spot suspicious lockstep behavior among groups of fraudulent user accounts who review suspiciously similar sets of products. 
Previous work \cite{shin2018fast,maruhashi2011multiaspectforensics,jiang2015general} has shown the benefit of incorporating extra information, such as timestamps, ratings, and review keywords, by modeling review data as a tensor. 
Tensors allow us to consider additional dimensions in order to identify suspicious behavior of interest more accurately and specifically. That is, extraordinarily dense subtensors indicate groups of users with lockstep behaviors both in the products they review and along the additional dimensions (e.g., multiple users reviewing the same products at the exact same time). 

In addition to review-fraud detection, spotting dense subtensors has been found effective for many anomaly-detection tasks. Examples include network-intrusion detection in TCP dumps \cite{shin2018fast,maruhashi2011multiaspectforensics}, retweet-boosting detection in online social networks \cite{jiang2015general}, bot-activity detection in Wikipedia \cite{shin2018fast}, and genetics applications \cite{saha2010dense,maruhashi2011multiaspectforensics}.

Due to these wide applications, several methods have been proposed for rapid and accurate dense-subtensor detection, and search-based methods have shown the best performance.
Specifically, search-based methods \cite{shin2018fast,jiang2015general} outperform methods based on tensor decomposition, such as CP Decomposition and HOSVD \cite{maruhashi2011multiaspectforensics}, in terms of accuracy and flexibility with regard to the choice of density metrics. Moreover, the latest search-based methods \cite{shin2018fast} provide a guarantee on the densities of the subtensors it finds, while methods based on tensor decomposition do not.

However, existing search methods for dense-subtensor detection assume that input tensors are small enough to fit in memory. 
Moreover, they are not directly applicable to tensors stored in disk since using them for such tensors incurs too many disk I/Os due to their highly iterative nature.
However, real applications, such as social media and web, often involve disk-resident tensors with terabytes or even petabytes, which in-memory algorithms cannot handle. This leaves a growing gap that needs to be filled.

\begin{table}[t]
	\centering{
		\caption{\label{tab:compare} Comparison of \method and state-of-the-art dense-subtensor detection methods. `\cmark' denotes `supported'.} 
		\begin{tabular}{c|ccccc|c}
			\toprule
			\ &  \rotatebox[origin=l]{90}{\parbox{2.5cm}{\mzoom \& \\ \mbiz \cite{shin2018fast}}} &  
			\rotatebox[origin=l]{90}{\parbox{2.8cm}{\densealert \& \densestream \cite{shin2017densealert}}} &  
			\rotatebox[origin=l]{90}{\parbox{3cm}{\cross \cite{jiang2015general}}}
			& \rotatebox[origin=l]{90}{\parbox{2cm}{MAF \cite{maruhashi2011multiaspectforensics}}}
			& \rotatebox[origin=l]{90}{\parbox{3cm}{\fraudar  \cite{hooi2017graph}}}  & \rotatebox[origin=l]{90}{\parbox{3cm}{\bf \method \\ (Proposed)}} \\
			\midrule
			High-order Tensors  & \cmark & \cmark & \cmark & \cmark  & & {\large \cmark} \\
			Flexibility in Density Measures  & \cmark & & \cmark & & \cmark & {\large \cmark} \\
			Accuracy Guarantees & \cmark & \cmark & & & \cmark & {\large \cmark} \\
			Out-of-core Computation &  & &  & & & {\large \cmark}  \\
			Distributed Computation  &  & & & & & {\large \cmark} \\
			\bottomrule
		\end{tabular}
	}
\end{table}

\subsection{Our Contributions}

To overcome these limitations, we propose \method 
a dense-subtensor detection method for disk-resident tensors. \method works under the W-Stream model \cite{ruhl2003efficient}, where data are only sequentially read and written during computation. As seen in Table \ref{tab:compare}, only \method supports out-of-core computation, which allows it to process data too large to fit in main memory. 
\method is optimized for this setting by carefully minimizing the amount of disk I/O and the number of steps requiring disk accesses, without losing accuracy guarantees it provides.
Moreover, we present a distributed version of \method using the \mapreduce framework \cite{dean2008mapreduce}, specifically its open source implementation \hadoop. 

The main strengths of \method are summarized as follows:
\begin{itemize}
	\item {\bf Memory Efficient:} 
	\method requires up to {\it 1,561$\times$} less memory and successfully handles {\it 1,000$\times$} larger data ({\it 2.6TB}) than its best competitors (Figures~\ref{fig:crown}A and \ref{fig:crown}B).
	\item {\bf Fast:} \method detects dense subtensors up to {\it 7$\times$} faster in real-world tensors and {\it 12 $\times$} faster in synthetic tensors than than its best competitors due to its near-linear scalability with all aspects of tensors (Figure~\ref{fig:crown}A).
	\item {\bf Provably Accurate:}  \method provides a guarantee on the densities of the subtensors it finds (Theorem~\ref{thm:accuracy:guarantee}), and it shows similar or higher accuracy in dense-subtensor detection than its best competitors on real-world tensors  (Figure~\ref{fig:crown}B).
	\item {\bf Effective:} \method successfully spotted network attacks from TCP dumps, and lockstep behavior in rating data, with the highest accuracy (Figure~\ref{fig:crown}C).
\end{itemize}

\noindent{\bf Reproducibility:} The code and data used in the paper are available at {\bf \url{http://dmlab.kaist.ac.kr/dcube}}.

\subsection{Related Work}
\label{sec:related}
We discuss previous work on (a) dense-subgraph detection, (b) dense-subtensor detection, (c) large-scale tensor decomposition, and (d) other anomaly/fraud detection methods.

{\bf Dense Subgraph Detection.} Dense-subgraph detection in graphs has been extensively studied in theory; see \cite{lee2010survey} for a survey. Exact algorithms \cite{goldberg1984finding,khuller2009finding} and approximate algorithms \cite{charikar2000greedy,khuller2009finding} have been proposed for finding subgraphs with maximum average degree. These have been extended for incorporating size restrictions \cite{andersen2009finding}, alternative metrics for denser subgraphs \cite{tsourakakis2013denser}, evolving graphs \cite{epasto2015efficient}, subgraphs with limited overlap \cite{balalau2015finding,galbrun2016top}, and streaming or distributed settings \cite{bahmani2012densest,bahmani2014efficient}.
Dense subgraph detection has been applied to fraud detection in social or review networks \cite{jiang2014catchsync,beutel2013copycatch,shah2014spotting,hooi2017graph,shin2016corescope}.

{\bf Dense Subtensor Detection.} Extending dense subgraph detection to tensors \cite{jiang2015general,shin2017densealert,shin2018fast} incorporates additional dimensions, such as time, to identify dense regions of interest with greater accuracy and specificity. 
\cross \cite{jiang2015general}, which starts from a seed subtensor and adjusts it in a greedy way until it reaches a local optimum, shows high accuracy in practice but does not provide any theoretical guarantees on its running time and accuracy.
\mzoom \cite{shin2018fast}, which starts from the entire tensor and only shrinks it by removing attributes one by one in a greedy way, improves \cross in terms of speed and approximation guarantees.
\change{\mbiz \cite{shin2018fast} starts from the output of \mzoom and repeats adding or removing an attribute greedily until a local optimum is reached.}
Given a dynamic tensor, \densestream and \densealert incrementally compute a single dense subtensor in it \cite{shin2017densealert}.
\cross, \mzoom, \change{\mbiz}, and \densestream require all tuples of relations to be loaded into memory at once and to be randomly accessed, which limit their applicability to large-scale datasets. 
\densealert maintains only the tuples created within a time window, and thus it can find a dense subtensor only within the  window.
Dense-subtensor detection in tensors has been found useful for detecting retweet boosting \cite{jiang2015general}, network attacks \cite{shin2018fast,shin2017densealert,maruhashi2011multiaspectforensics}, bot activities \cite{shin2018fast}, and vandalism on Wikipedia \cite{shin2017densealert}, and also for genetics applications \cite{saha2010dense,maruhashi2011multiaspectforensics}. 

{\bf Large-Scale Tensor Decomposition.} Tensor decomposition such as HOSVD and CP decomposition \cite{kolda2009tensor} can be used to spot dense subtensors \cite{maruhashi2011multiaspectforensics}. 
Scalable algorithms for tensor decomposition have been developed, including disk-based algorithms \cite{oh2017shot,shin2014distributed}, distributed algorithms \cite{kang2012gigatensor,shin2014distributed,jeon2015haten2}, and approximate algorithms based on sampling \cite{papalexakis2012parcube} and count-min sketch \cite{wang2015fast}.
However, dense-subtensor detection based on tensor decomposition has serious limitations: it usually detects subtensors with significantly lower density (see Section~\ref{sec:experiments:speed}) than search-based methods, provides no flexibility with regard to the choice of density metric, and does not provide any approximation guarantee. 

{\bf Other Anomaly/Fraud Detection Methods.}
In addition to dense-subtensor detection, many approaches, including those based on egonet features \cite{akoglu2010oddball}, coreness \cite{shin2016corescope}, and behavior models \cite{rossi2013modeling}, have been used for anomaly and fraud detection in graphs.
See \cite{akoglu2015graph} for a survey.

\subsection{Organization of the Paper}
In Section~\ref{sec:prelim}, we provide notations and a formal problem definition.
In Section~\ref{sec:method}, we propose \method, a disk-based dense-subtensor detection method.
In Section~\ref{sec:experiments}, we present experimental results and discuss them.
In Section~\ref{sec:conclusion}, we offer conclusions.

\section{Preliminaries and Problem Definition}
\label{sec:prelim}

\begin{table}[!t]
	\centering
	\caption{Table of symbols.}
	\begin{tabular}{r|l}
		\toprule
		\textbf{Symbol} & \textbf{\qquad\qquad\qquad Definition} \\
		\midrule
		$\TR(A_{1},\cdots,A_{N},X)$ & relation representing an $N$-way tensor\\
		$N$ & number of the dimension attributes in $\TR$\\
		$A_{n}$ & $n$-th dimension attribute in $\TR$ \\
		$X$ & measure attribute in $\TR$ \\
		$t[A_{n}]$ (or $t[X]$)& value of attribute $A_{n}$ (or $X$) in tuple $t$ in $\TR$ \\
		$\TB$ & a subtensor in $\TR$ \\
		$\densityno{\TB}{\TR}$ & density of subtensor $\TB$ in $\TR$ \\
		$\TR_{n}$ (or $\TB_{n}$) & set of distinct values of $A_{n}$ in $\TR$ (or $\TB$)\\
		$\mass{\TR}$ (or $\mass{\TB}$) & mass of $\TR$ (or $\TB$) \\ 
		$\TB(a,n)$ & set of tuples with attribute $A_{n}= a$ in $\TB$ \\
		$\mass{\TB(a,n)}$ & attribute-value mass of $a$ in $A_{n}$\\ 
		$k$ & number of subtensors we aim to find \\
		$\theta$ & mass-threshold parameter in \method \\
		$[x]$ & $\{1,2,\cdots,x\}$ \\
		\bottomrule
	\end{tabular}
	\label{tab:symbols}
\end{table}

In this section, we first introduce notations and concepts used in the paper.
Then, we define density measures and the problem of top-$k$ dense-subtensor detection. \\

\subsection{Notations and Concepts}
\label{sec:prelim:notation}
Table~\ref{tab:symbols} lists the symbols frequently used in the paper.
We use $[x]=\{1,2,\cdots,x\}$ for brevity.
Let $\TR(A_{1},\cdots,A_{N},X)$ be a relation with $N$ dimension attributes, denoted by $A_{1},\cdots,A_{N}$, and a nonnegative measure attribute, denoted by $X$ (see Example~\ref{example} for a running example).
For each tuple $t\in\TR$ and for each $n\in[N]$, $t[A_{n}]$ and $t[X]$ indicate the values of $A_{n}$ and $X$, resp., in $t$.
For each $n\in[N]$, we use $\TR_{n}=\{t[A_{n}]:t\in \TR\}$ to denote the set of distinct values of $A_{n}$ in $\TR$. 
The relation $\TR$ is naturally represented as an $N$-way tensor of size $|\TR_{1}|\times\cdots\times|\TR_{N}|$.
The value of each entry in the tensor is $t[X]$, if the corresponding tuple $t$ exists, and $0$ otherwise.
Let $\TB_{n}$ be a subset of $\TR_{n}$.
Then, a {\it subtensor} $\TB$ in $\TR$ is defined as $\TB(A_{1},...,A_{N}, X)=\{t\in\TR : \forall n \in [N], t[A_{n}]\in \TB_{n}\}$, the set of tuples where each attribute $A_{n}$ has a value in $\TB_{n}$.
The relation $\TB$ is a `subtensor' because it forms a subtensor of size $|\TB_{1}|\times\cdots\times|\TB_{N}|$ in the tensor representation of $\TR$, as in Figure~\ref{fig:example}B.
We define the mass of $\TR$ as $\mass{\TR}=\sum_{t\in\TR}t[X]$, the sum of attribute $X$ in the tuples of $\TR$.
We denote the set of tuples of $\TB$ whose attribute $A_{n}=a$ by $\TB(a,n)=\{t\in\TB : t[A_{n}] = a\}$ and its mass, called the {\it attribute-value mass of $a$ in $A_{n}$}, by $\mass{\TB(a,n)}=\sum_{t\in\TB(a,n)}t[X]$.

\begin{figure}[t]
	\centering
	\begin{tabular}{cc}
		\begin{minipage}{.5\textwidth}
			\center
			\includegraphics[width=0.8\linewidth]{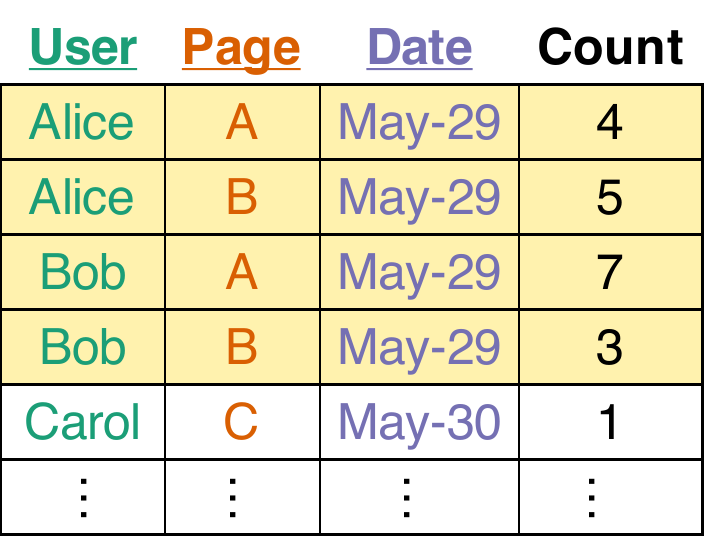}
		\end{minipage} 
		& \begin{minipage}{.5\textwidth}
			\center
			\includegraphics[width=\linewidth]{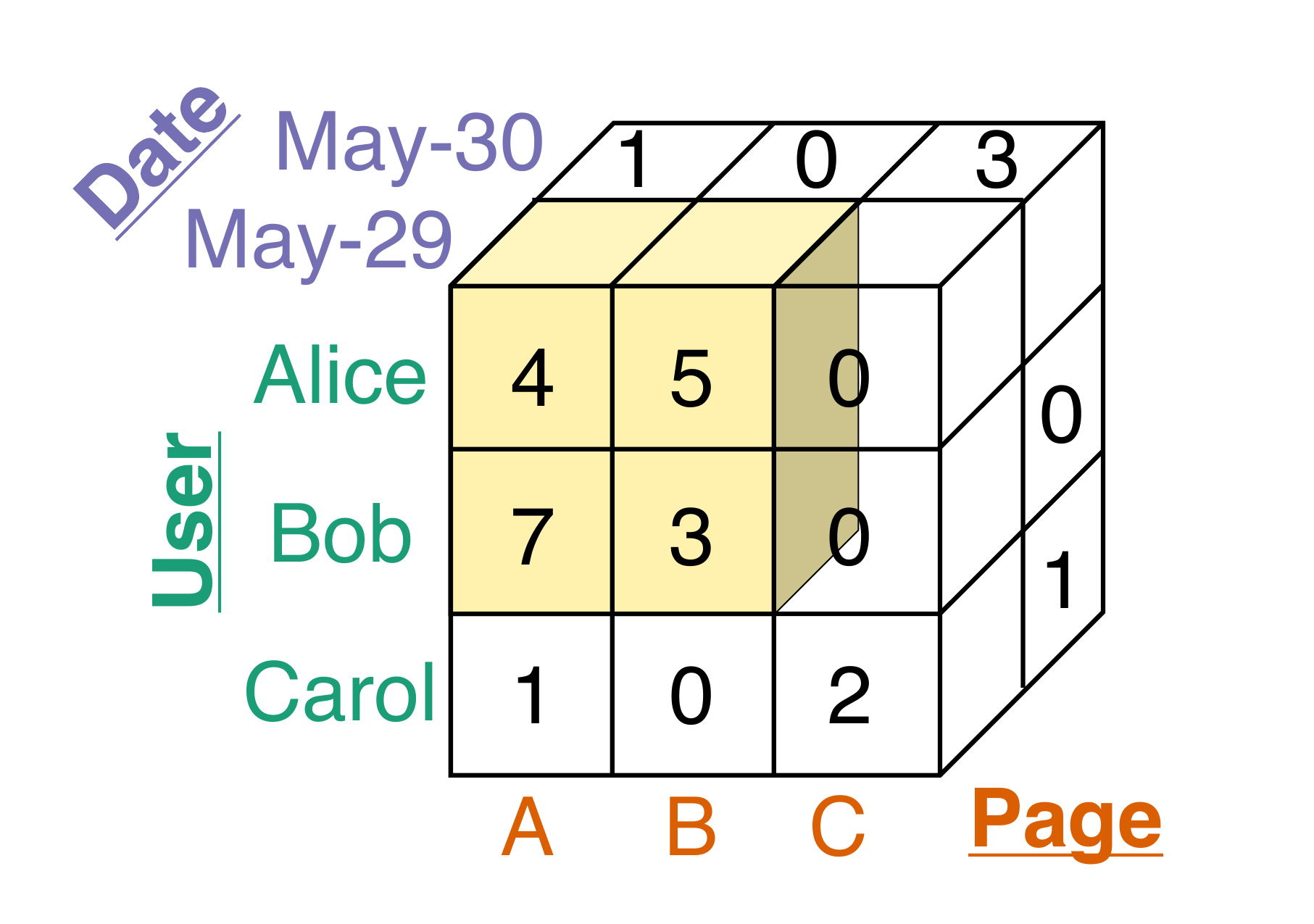}
		\end{minipage} \vspace{1mm} \\
		A. Relation $\TR$ & B. Tensor Representation of $\TR$ \\
	\end{tabular}
	\caption{\label{fig:example}
		Pictorial description of Example~\ref{example}. A. Relation $\TR$ where the colored tuples compose relation $\TB$. B. Tensor representation of $\TR$ where the relation $\TB$ forms a subtensor.
	}
\end{figure}

\begin{example}[Wikipedia Revision History] \label{example} As in Figure~\ref{fig:example},
	assume a relation $\TR(\underline{user}, \underline{\smash{page}}, \underline{date}, count)$, where each tuple $(u, p, d, c)$ in $\TR$ indicates that user $u$ revised page $p$, $c$ times, on date $d$.
	The first three attributes, $A_{1}$=$user$, $A_{2}$=$page$,
	and $A_{3}$=$date$, are dimension attributes, and the other one, $X$=$count$, is the
	measure attribute.
	Let $\TB_{1}$=$\{Alice, Bob\}$, $\TB_{2}$=$\{A,B\}$, and $\TB_{3}$=$\{May$-$29\}$.
	Then, $\TB$ is the set of tuples regarding the revision of page A or B by Alice or Bob on May-29, and its mass $\mass{\TB}$ is $19$, the total number of such revisions.
	The attribute-value mass of Alice (i.e., $\mass{\TB(Alice, 1)}$)  is $9$, the number of revisions on A or B by exactly Alice on May-29. 
	In the tensor representation, $\TB$ composes a subtensor in $\TR$, as depicted in Figure~\ref{fig:example}B.
\end{example}

\subsection{Density Measures}
\label{sec:prelim:density}

We present density measures proven useful for anomaly detection in past studies.
We use them throughout the paper although our dense-subtensor detection method, explained in Section~\ref{sec:method}, is flexible and not restricted to specific measures. 
Below, we slightly abuse notations to emphasize that the density measures are the functions of $\mass{\TB}$, $\{|\TB_{n}|\}_{n\text{=}1}^{N}$, $\mass{\TR}$, and $\{|\TR_{n}|\}_{n\text{=}1}^{N}$, where $\TB$ is a subtensor of a relation $\TR$.

Arithmetic Average Mass (Definition~\ref{defn:density:ari}) and Geometric Average Mass (Definition~\ref{defn:density:geo}), which were used for detecting network intrusions and bot activities \cite{shin2018fast}, are the extensions of density measures widely-used for graphs \cite{charikar2000greedy,kannan1999analyzing}.

\begin{definition}[Arithmetic Average Mass $\densityarinoarg$ \cite{shin2018fast}] \label{defn:density:ari} The {\it arithmetic average mass} of a subtensor $\TB$ of a relation $\TR$ is defined as 
	\begin{equation*}
	\densityari{\TB}{\TR}\text{\normalfont=}\densityarinoarg(\mass{\TB}, \{|\TB_{n}|\}_{n\text{=}1}^{N}, \mass{\TR}, \{|\TR_{n}|\}_{n\text{=}1}^{N})=\frac{\mass{\TB}}{\frac{1}{N}\sum_{n\text{=}1}^{N}|\TB_{n}|}.
	\end{equation*}
\end{definition}
\begin{definition}[Geometric Average Mass $\densitygeonoarg$ \cite{shin2018fast}] The {\it geometric average mass} of a subtensor $\TB$ of a relation $\TR$ is defined as  \label{defn:density:geo}
	\begin{equation*}
	\densitygeo{\TB}{\TR}\text{\normalfont=}\densitygeonoarg(\mass{\TB}, \{|\TB_{n}|\}_{n\text{=}1}^{N}, \mass{\TR}, \{|\TR_{n}|\}_{n\text{=}1}^{N})=\frac{\mass{\TB}}{(\prod_{n\text{=}1}^{N}|\TB_{n}|)^{\frac{1}{N}}}.
	\end{equation*}
\end{definition}

Suspiciousness (Definition~\ref{defn:density:susp}), which was used for detecting `retweet-boosting' activities \cite{jiang2014catchsync}, is the negative log-likelihood that $\TB$ has mass $\mass{\TB}$ under the assumption that each entry of $\TR$ is i.i.d from a Poisson distribution.
\begin{definition}[Suspiciousness $\densitysuspnoarg$ \cite{jiang2015general}] The {\it suspiciousness} of a subtensor $\TB$ of a relation $\TR$ is defined as  \label{defn:density:susp}
	\begin{align*}
	\densitysusp{\TB}{\TR} & = \densitysuspnoarg(\mass{\TB}, \{|\TB_{n}|\}_{n=1}^{N}, \mass{\TR}, \{|\TR_{n}|\}_{n=1}^{N}) \\
	& = \mass{\TB}\left(\log\frac{\mass{\TB}}{\mass{\TR}}-1\right)+\mass{\TR}\prod_{n=1}^{N}\frac{|\TB_{n}|}{|\TR_{n}|}-\mass{\TB}\log\left(\prod_{n=1}^{N}\frac{|\TB_{n}|}{|\TR_{n}|}\right).
	\end{align*}
\end{definition}

Entry Surplus (Definition~\ref{defn:surp}) is the observed mass of $\TB$ subtracted by $\alpha$ times the expected mass, under the assumption that the value of each entry (in the tensor representation) in $\TR$ is i.i.d. 
It is a multi-dimensional extension of edge surplus \cite{tsourakakis2013denser}, a density metric for graphs.
\begin{definition}[Entry Surplus \cite{shin2018fast}]
	\label{defn:surp}
	The {\it entry surplus} of a subtensor $\TB$ of a relation $\TR$ is defined as 
	\begin{align*}
	\densitysurp{\TB}{\TR} & = \densitysurpnoarg(\mass{\TB}, \{|\TB_{n}|\}_{n=1}^{N}, \mass{\TR}, \{|\TR_{n}|\}_{n=1}^{N})  \\
	& = \mass{\TB}-\alpha\mass{\TR}\prod_{n=1}^{N}\frac{|\TB_{n}|}{|\TR_{n}|}.
	\end{align*}
\end{definition}
Subtensors with high entry surplus are configurable by adjusting $\alpha$.
With high $\alpha$ values, relatively small compact subtensors have higher entry surplus than large sparse subtensors, while the opposite happens with small $\alpha$ values.
We show this tendency experimentally in Section~\ref{sec:exp:surp}.

\subsection{Problem Definition}
\label{sec:prelim:problem}

Based on the concepts and density measures in the previous sections, we define the problem of top-$k$ dense-subtensor detection in a large-scale tensor in Definition~\ref{defn:problem}.
\begin{problem}[Large-scale Top-$k$ Densest Subtensor Detection] \label{defn:problem}
	{\bf (1) Given:} a large-scale relation $\TR$ not fitting in memory, the number of subtensors $k$, and a density measure $\density$,
	{\bf (2) Find:} the top-$k$ subtensors of $\TR$ with the highest density in terms of $\density$.
\end{problem}

Even when we restrict our attention to finding one subtensor in a matrix fitting in memory (i.e., $k=1$ and $N=2$), obtaining an exact solution takes $O((\sum_{n=1}^{N}|\TR_{n}|)^{6})$ time \cite{goldberg1984finding,khuller2009finding}, which is infeasible for large-scale tensors.
Thus, our focus in this work is to design an approximate algorithm with  (1) near-linear scalability with all aspects of $\TR$, which does not fit in memory, (2) an approximation guarantee at least for some density measures, and (3) meaningful results on real-world data.

\section{Proposed Method}
\label{sec:method}
In this section, we propose \method, a disk-based dense-subtensor detection method. 
We first describe \method in Section~\ref{sec:method:alg}. Then, we prove its theoretical properties in Section~\ref{sec:method:analysis}.
Lastly, we present our \mapreduce implementation of \method in Section~\ref{sec:method:mr}.
Throughout these subsections, we assume that the entries of tensors (i.e., the tuples of relations) are stored on disk and read/written only in a sequential way. 
However, all other data (e.g., distinct attribute-value sets and the mass of each attribute value) are assumed to be stored in memory.

\subsection{Algorithm}
\label{sec:method:alg}
\method is a search method that starts with the given relation and removes attribute values (and the tuples with the attribute values) sequentially so that a dense subtensor is left.
Contrary to previous approaches, \method removes multiple attribute values (and the tuples with the attribute values) at a time to reduce the number of iterations and also disk I/Os.
In addition to this advantage, \method carefully chooses attribute values to remove to give the same accuracy guarantee as if attribute values were removed one by one, and shows similar or even higher accuracy empirically.

\begin{algorithm}[t]
	\caption{\method}
	\label{alg:method}
	\small
	\SetKwInOut{Input}{Input}
	\SetKwInOut{Output}{Output}
	\Input{
		relation: $\TR$,
		\ density measure: $\density$, \\
		\ threshold: $\theta(\geq 1)$, \\
		\ the number of subtensors we aim to find: $k$
	}
	\Output{
		$k$ dense subtensors
	}
	$\TR^{ori} \leftarrow copy(\TR)$ \label{alg:method:line:copy} \\
	compute $\{\TR_{n}\}_{n=1}^{N}$ \label{alg:method:line:distinct}\\
	$results\leftarrow \emptyset$ \bcomment{list of dense subtensors} \\
	\For{$i\leftarrow 1..k$}{
		$\mass{\TR} \leftarrow \sum_{t\in\TR}t[X]$ \label{alg:method:line:mass} \\
		$\{\TB_{n}\}_{n=1}^{N}\leftarrow find\_one(\TR, \{\TR_{n}\}_{n=1}^{N}, \mass{\TR},  \density, \theta)$ 
		\label{alg:method:line:single} \qquad \qquad
		\bcomment{see Algorithm~\ref{alg:method:single}}\\
		${\TR} \leftarrow \{t \in \TR: \exists n\in[N], t[A_{n}]\notin \TB_{n}\}$ \label{alg:method:line:remove} \bcomment{$\TR \leftarrow \TR - \TB$} \\
		${\TB}^{ori}\leftarrow\{t \in \TR^{ori} : \forall n\in[N], t[A_{n}]\in \TB_{n}\}$ \label{alg:method:line:refine}\\
		$results\leftarrow results\cup \{\TB^{ori}\}$  \label{alg:method:line:single:include}\\
	}
	{\bf return} $results$
\end{algorithm}

\subsubsection{Overall Structure of D-Cube (Algorithm~\ref{alg:method})}
\label{sec:method:alg:overall}

Algorithm~\ref{alg:method} describes the overall structure of \method. It first copies and assigns the given relation $\TR$ to $\TR^{ori}$ (line~\ref{alg:method:line:copy}); and computes the sets of distinct attribute values composing $\TR$ (line~\ref{alg:method:line:distinct}).
Then, it finds $k$ dense subtensors one by one from $\TR$ (line~\ref{alg:method:line:single}) using its mass as a parameter (line~\ref{alg:method:line:mass}).
The detailed procedure for detecting a single dense subtensor from $\TR$ is explained in Section~\ref{sec:method:alg:single}.
After each subtensor $\TB$ is found, the tuples included in $\TB$ are removed from $\TR$ (line~\ref{alg:method:line:remove}) to prevent the same subtensor from being found again.
Due to this change in $\TR$, subtensors found from $\TR$ are not necessarily the subtensors of the original relation $\TR^{ori}$.
Thus, instead of $\TB$, the subtensor in $\TR^{ori}$ formed by the same attribute values forming $\TB$ is added to the list of $k$ dense subtensors (lines~\ref{alg:method:line:refine}-\ref{alg:method:line:single:include}).
Notice that, due to this step, \method can detect overlapping dense subtensors. 
That is, a tuple can be included in multiple dense subtensors. 

Based on our assumption that the sets of distinct attribute values (i.e., $\{\TR_{n}\}_{n=1}^{N}$ and $\{\TB_{n}\}_{n=1}^{N}$) are stored in memory and can be randomly accessed, all the steps in Algorithm~\ref{alg:method} can be performed by sequentially reading and writing tuples in relations (i.e.,  tensor entries) in disk without loading all the tuples in memory at once.
For example, the filtering steps in lines~\ref{alg:method:line:remove}-\ref{alg:method:line:refine} can be performed by sequentially reading each tuple from disk and writing the tuple to disk only if it satisfies the given condition.

Note that this overall structure of \method is similar to that of \mzoom \cite{shin2018fast} except that tuples are stored on disk.
However, the methods differ significantly in the way each dense subtensor is found from $\TR$, which is explained in the following section.

\begin{algorithm}[t]
	\caption{$find\_one$ in \method}
	\label{alg:method:single}
	\small
	\SetKwInOut{Input}{Input}
	\SetKwInOut{Output}{Output}
	\Input{
		relation: $\TR$, \ attribute-value sets: $\{\TR_{n}\}_{n=1}^{N}$,\\
		\ mass: $\mass{\TR}$, density measure: $\density$, threshold: $\theta(\geq 1)$
	}
	\Output{
		attribute values forming 
		a dense subtensor \\
	}
	$\TB \leftarrow copy(\TR)$, $\mass{\TB}\leftarrow\mass{\TR}$ \label{alg:method:single:line:init1} \bcomment{initialize the subtensor $\TB$} \\
	$\TB_{n}\leftarrow copy(\TR_{n})$, $\forall n\in[N]$ \label{alg:method:single:line:init2}\\
	$\tilde{\density}\leftarrow\density(\mass{\TB}, \{|\TB_{n}|\}_{n=1}^{N}, \mass{\TR}, \{|\TR_{n}|\}_{n=1}^{N})$ \bcomment{$\tilde{\rho}$: max $\density$ so far}  \\
	$r, \tilde{r}\leftarrow 1$ \bcomment{$r$: current order of attribute values, $\tilde{r}$: $r$ with $\tilde{\density}$} \\
	\While{$\exists n\in [N], \TB_{n}\neq \emptyset$}{ \label{alg:method:single:loop} \nonl \vspace{-2.5mm} \  \bcomment{until all values are removed} \\
		compute $\{\{\mass{\TB(a,n)}\}_{a\in\TB_{n}}\}_{n=1}^{N}$ \label{alg:method:line:single:mass} \\
		$i \leftarrow select\_dimension()$ \bcomment{see Algorithms~\ref{alg:method:guaranteed} and \ref{alg:method:flexible} \label{alg:method:line:single:select}}\\ 
		$D_{i}\leftarrow \{a\in \TB_{i} : \mass{\TB(a,i)} \leq \theta\frac{\mass{\TB}}{|\TB_{i}|}\}$ \label{alg:method:line:single:threshold} \bcomment{$D_{i}$: set to be removed} \\
		sort $D_{i}$ in an increasing order of $\mass{\TB(a,i)}$ \label{alg:method:line:single:sort} \\
		\For{{\normalfont each value} $a\in D_{i}$}{
			$\TB_{i}\leftarrow \TB_{i}-\{a\}$, $\mass{\TB}\leftarrow \mass{\TB}-\mass{\TB(a,i)}$ \label{alg:method:line:single:remove_one} \\
			$\density'\leftarrow \density(\mass{\TB}, \{|\TB_{n}|\}_{n=1}^{N}, \mass{\TR}, \{|\TR_{n}|\}_{n=1}^{N})$  \label{alg:method:line:single:density} \\ \nonl \ \bcomment{$\density':$ $\density$ when $a$ is removed} \\
			$order(a,i)\leftarrow r$, $r\leftarrow r+1$ \label{alg:method:line:single:order1} \\
			\If{$\density' > \tilde{\density}$}{
				$\tilde{\density}\leftarrow\density'$, $\tilde{r}\leftarrow r$
				\label{alg:method:line:single:order2}  \bcomment{update max $\density$ so far}
			}
		}
		$\TB\leftarrow \{t\in\TB : t[A_{i}]\notin D_{i}\}$ \label{alg:method:line:single:remove} \bcomment{remove tuples} \\
	}
	$\tilde{\TB}_{n}\leftarrow\{a\in\TR_{n}: order(a,n) \geq \tilde{r}\}$, $\forall n\in [N]$ \bcomment{reconstruct}
	\label{alg:method:line:single:reconstruct}\\
	{\bf return} $\{\tilde{\TB}_{n}\}_{n=1}^{N}$ \label{alg:method:line:single:return}
\end{algorithm}

\subsubsection{Single Subtensor Detection (Algorithm~\ref{alg:method:single})}
\label{sec:method:alg:single}

Algorithm~\ref{alg:method:single} describes how \method detects each dense subtensor from the given relation $\TR$.
It first initializes a subtensor $\TB$ to $\TR$ (lines~\ref{alg:method:single:line:init1}-\ref{alg:method:single:line:init2}) then repeatedly removes attribute values and the tuples of $\TB$ with those attribute values until all values are removed (line~\ref{alg:method:single:loop}).

Specifically, in each iteration, \method first chooses a dimension attribute $A_{i}$ that attribute values are removed from (line~\ref{alg:method:line:single:select}).
Then, it computes $D_{i}$, the set of attribute values whose masses are less than $\theta(\geq 1)$ times the average (line~\ref{alg:method:line:single:threshold}).
We explain how the dimension attribute is chosen, in Section~\ref{sec:method:alg:dimension} and
analyze the effects of $\theta$ on the accuracy and the time complexity, in Section~\ref{sec:method:analysis}.
The tuples whose attribute values of $A_{i}$ are in $D_{i}$ are removed from $\TB$ at once within a single scan of $\TB$ (line~\ref{alg:method:line:single:remove}).
However, deleting a subset of $D_{i}$ may achieve higher value of the metric $\density$. Hence, \method computes the changes in the density of $\TB$ (line 11) as if the attribute values in $D_i$ were removed one by one, in an increasing order of their masses. This allows \method to optimize $\density$ as if we removed attributes one by one, while still benefiting from the computational speedup of removing multiple attributes in each scan.
Note that these changes in $\density$ can be computed exactly without actually removing the tuples from $\TB$ or even accessing the tuples in $\TB$ since its mass (i.e., $\mass{\TB}$) and the number of distinct attribute values (i.e., $\{|\TB_{n}|\}_{n=1}^{N}$) are maintained up-to-date (lines~\ref{alg:method:line:single:remove_one}-\ref{alg:method:line:single:density}).
This is because removing an attribute value from a dimension attribute does not affect the masses of the other values of the same attribute.
The orders that attribute values are removed and when the density of $\TB$ is maximized are maintained (lines~\ref{alg:method:line:single:order1}-\ref{alg:method:line:single:order2}) so that the subtensor $\TB$ maximizing the density can be restored and returned (lines~\ref{alg:method:line:single:reconstruct}-\ref{alg:method:line:single:return}), as the result of Algorithm~\ref{alg:method:single}.

Note that, in each iteration (lines~\ref{alg:method:single:loop}-\ref{alg:method:line:single:remove}) of Algorithm~\ref{alg:method:single}, the tuples of $\TB$, which are stored on disk, need to be scanned only twice, once in line~\ref{alg:method:line:single:mass} and once in line~\ref{alg:method:line:single:remove}.
Moreover, both steps can be performed by simply sequentially reading and/or writing tuples in $\TB$ without loading all the tuples in memory at once.
For example, to compute attribute-value masses in line~\ref{alg:method:line:single:mass}, \method increases $\mass{\TB(t[A_{n}],n)}$ by $t[X]$ for each dimension attribute $A_{n}$ after reading each tuple $t$ in $\TB$ sequentially from disk.

\begin{algorithm}[t]
	\caption{$select\_dimension$ by cardinality}
	\label{alg:method:guaranteed}
	\small
	\SetKwInOut{Input}{Input}
	\SetKwInOut{Output}{Output}
	\Input{
		attribute-value sets: $\{\TB_{n}\}_{n=1}^{N}$
	}
	\Output{
		a dimension in $[N]$
	}
	{\bf return} $n\in[N]$ with maximum $|\TB_{n}|$
\end{algorithm}

\begin{algorithm}[t]
	\caption{$select\_dimension$ by density}
	\label{alg:method:flexible}
	\small
	\SetKwInOut{Input}{Input}
	\SetKwInOut{Output}{Output}
	\Input{
		attribute-value sets: $\{\TB_{n}\}_{n=1}^{N}$ and $\{\TR_{n}\}_{n=1}^{N}$, \\
		\ {\small attribute-value masses: $\{\{\mass{\TB(a,n)}\}_{a\in\TB_{n}}\}_{n=1}^{N}$},\\
		\ masses: $\mass{\TB}$ and $\mass{\TR}$, density measure: $\density$, \\
		\ threshold: $\theta(\geq 1)$
	}
	\Output{
		a dimension in $[N]$
	}
	$\tilde{\density}\leftarrow -\infty$, $\tilde{i}\leftarrow 1$  \bcomment{$\tilde{\rho}$: max $\density$ so far, $\tilde{i}$: dimension with  $\tilde{\rho}$} \\
	\For{{\normalfont each dimension} $i\in[N]$}{
		\If{$\TB_{i}\neq\emptyset$}{
			$D_{i}\leftarrow \{a\in \TB_{i} : \mass{\TB(a,i)} \leq \theta\frac{\mass{\TB}}{|\TB_{i}|}\}$ \\ \nonl \ \bcomment{$D_{i}:$ set to be removed} \\
			$\mass{\TB}'\leftarrow \mass{\TB}-\sum_{a\in D_{i}}\mass{\TB(a,i)}$ \\
			$\TB'_{i} \leftarrow \TB_{i}-D_{i}$ \\
			$\density'\leftarrow \density(\mass{\TB}', \{|\TB_{n}|\}_{n\neq i}\cup\{|\TB'_{i}|\}, \mass{\TR}, \{|\TR_{n}|\}_{n=1}^{N})$ \\ \nonl \ \bcomment{$\density':$ $\density$ when $D_{i}$ are removed} \\
			\If{$\density' > \tilde{\density}$}{
				$\tilde{\density}\leftarrow \density'$, $\tilde{i}\leftarrow i$\bcomment{update max $\density$ so far}
			}
		}
	}
	{\bf return} $\tilde{i}$
\end{algorithm}

\subsubsection{Dimension Selection (Algorithms~\ref{alg:method:guaranteed} and \ref{alg:method:flexible})}
\label{sec:method:alg:dimension}

We discuss two policies for choosing a dimension attribute that attribute values are removed from.
They are used in line~\ref{alg:method:line:single:select} of Algorithm~\ref{alg:method:single} offering different advantages.

{\bf Maximum Cardinality Policy (Algorithm~\ref{alg:method:guaranteed})}: The dimension attribute with the largest cardinality is chosen, as described in Algorithm~\ref{alg:method:guaranteed}.
This simple policy, however, provides an accuracy guarantee (see Theorem~\ref{thm:accuracy:guarantee} in Section~\ref{sec:method:analysis:accuracy}).

{\bf Maximum Density Policy (Algorithm~\ref{alg:method:flexible})}: The density of $\TB$ when attribute values are removed from each dimension attribute is computed. Then, the dimension attribute leading to the highest density is chosen.
Note that the tuples in $\TB$, stored on disk, do not need to be accessed for this computation, as described in Algorithm~\ref{alg:method:flexible}.
Although this policy does not provide the accuracy guarantee given by the maximum cardinality policy, this policy works well with various density measures and tends to spot denser subtensors than the maximum cardinality policy in our experiments with real-world data.

\subsubsection{Efficient Implementation}
\label{sec:method:alg:impl}
We present the optimization techniques used for the efficient implementation of \method.

{\bf Combining Disk-Accessing Steps.}
The amount of disk I/O can be reduced by combining multiple steps involving disk accesses.
In Algorithm~\ref{alg:method}, updating $\TR$ (line~\ref{alg:method:line:remove}) in an iteration can be combined with computing the mass of $\TR$ (line~\ref{alg:method:line:mass}) in the next iteration.
That is, if we aggregate the values of the tuples of $\TR$ while they are written for the update, we do not need to scan $\TR$ again for computing its mass in the next iteration.
Likewise, in Algorithm~\ref{alg:method:single}, updating $\TB$ (line~\ref{alg:method:line:single:remove}) in an iteration can be combined with computing attribute-value masses (line~\ref{alg:method:line:single:mass}) in the next iteration.
This optimization reduces the amount of disk I/O in \method about 30\%.

{\bf Caching Tensor Entries in Memory.} 
Although we assume that tuples are stored on disk, storing them in memory up to the memory capacity speeds up \method up to $3$ times in our experiments (see Section~\ref{sec:experiments:scalability:data}).
We cache the tuples in $\TB$, which are more frequently accessed than those in $\TR$ or $\TR^{ori}$, in memory with the highest priority.

\subsection{Analyses}
\label{sec:method:analysis}

In this section, we prove the time and space complexities of \method and the accuracy guarantee provided by \method.
\change{Then, we theoretically compare \method with \mzoom and \mbiz \cite{shin2018fast}.}

\subsubsection{Complexity Analyses}
\label{sec:method:analysis:complexity}
Theorem~\ref{thm:time:worst} states the worst-case time complexity, which equals to the worst-case I/O complexity, of \method.

\begin{lemma}[Maximum Number of Iterations in Algorithm~\ref{alg:method:single}]
	\label{lemma:epsilon}
	Let $L=\max_{n\in[N]}|\TR_{n}|$. Then,
	the number of iterations (lines~\ref{alg:method:single:loop}-\ref{alg:method:line:single:remove}) in Algorithm~\ref{alg:method:single} is at most
	\begin{equation*}
	N\min(\log_{\theta}L,L).
	\end{equation*}
\end{lemma}
\begin{proof}
	In each iteration (lines~\ref{alg:method:single:loop}-\ref{alg:method:line:single:remove})
	of Algorithm~\ref{alg:method:single}, among the values of the chosen
	dimension attribute $A_{i}$, attribute values whose masses are at most $\theta\frac{\mass{\TB}}{|\TB_{i}|}$, where $\theta\geq 1$, are removed.
	The set of such attribute values is denoted by $D_{i}$. We will show that, if ${|\TB_{i}|}>0$, then 
	\begin{equation}
	|\TB_{i}\backslash D_{i}|<{|\TB_{i}|}/{\theta}.\label{eq:decrease:speed}
	\end{equation}
	Note that, when $|\TB_{i}\backslash D_{i}|=0$, Eq.~\eqref{eq:decrease:speed}
	trivially holds. When $|\TB_{i}\backslash D_{i}|>0$, $\mass{\TB}$
	can be factorized and lower bounded as 
	\begin{align*}
	\mass{\TB} =\sum\nolimits _{a\in\TB_{i}\backslash D_{i}}\mass{\TB(a,i)}+\sum\nolimits _{a\in D_{i}}\mass{\TB(a,i)}
	\geq\sum\nolimits _{a\in\TB_{i}\backslash D_{i}}\mass{\TB(a,i)}>|\TB_{i}\backslash D_{i}|\cdot\theta\frac{\mass{\TB}}{|\TB_{i}|},
	\end{align*}
	where the last strict inequality is from the definition of $D_{i}$
	and that $|\TB_{i}\backslash D_{i}|>0$. This strict inequality
	implies $\mass{\TB}>0$, and thus dividing both sides by $\theta\frac{\mass{\TB}}{|\TB_{i}|}$
	gives Eq.~\eqref{eq:decrease:speed}. Now, Eq.~\eqref{eq:decrease:speed}
	implies that the number of remaining values of the chosen attribute
	after each iteration is less than $1/\theta$ of that before the iteration.
	Hence each attribute can be chosen at most $\log_{\theta}L$ times
	before all of its values are removed. Thus, the maximum number of
	iterations is at most $N\log_{\theta}L$. Also, by Eq.~\eqref{eq:decrease:speed},
	at least one attribute value is removed per iteration. Hence, the
	maximum number of iterations is at most the number of attribute values,
	which is upper bounded by $NL$. Hence the number of iterations is
	upper bounded by $N\max(\log_{\theta}L,L)$. 
\end{proof}

\begin{theorem}[Worst-case Time Complexity]
	\label{thm:time:worst}
	Let $L=\max_{n\in[N]}|\TR_{n}|$. 
	If $\theta=O\left(e^{(\frac{N|\TR|}{L})}\right)$, which is a weaker condition than $\theta=O(1)$, the worst-case time complexity of Algorithm~\ref{alg:method} is
	\begin{equation}
	O(kN^2|\TR|\min(\log_{\theta}L, L)). \label{eq:time_complexity}
	\end{equation}
\end{theorem}
\begin{proof}
	From Lemma~\ref{lemma:epsilon}, the number of iterations (lines~\ref{alg:method:single:loop}-\ref{alg:method:line:single:remove}) in Algorithm~\ref{alg:method:single} is $O(N\min(\log_{\theta}L,L))$.
	Executing lines~\ref{alg:method:line:single:mass} and \ref{alg:method:line:single:remove} $O(N\min(\log_{\theta}L,L))$ times takes $O(N^2 |\TR|\min(\log_{\theta}L,L))$, which dominates the time complexity of the other parts. For example, repeatedly executing line~\ref{alg:method:line:single:sort} takes $O(NL\log_{2}L)$, and by our assumption, it is dominated by $O(N^2 |\TR|\min(\log_{\theta}L,L))$.
	Thus, the worst-case time complexity of Algorithm~\ref{alg:method:single} is $O(N^2 |\TR|\min(\log_{\theta}L,L))$, and that of Algorithm~\ref{alg:method}, which executes Algorithm~\ref{alg:method:single}, $k$ times, is $O(kN^2 |\TR|\min(\log_{\theta}L,L))$. 
\end{proof}


However, this worst-case time complexity, which allows the worst distributions of the measure attribute values of tuples, is too pessimistic. 
In Section~\ref{sec:experiments:scalability:data}, we experimentally show that \method scales linearly with $k$, $N$, and $\TR$; and sub-linearly with $L$ even when $\theta$ is its smallest value $1$.


Theorem~\ref{thm:space} states the memory requirement of \method. 
Since the tuples do not need to be stored in memory all at once in \method, its memory requirement does not depend on the number of tuples (i.e., $|\TR|$).

\begin{theorem}[Memory Requirements]
	\label{thm:space}
	The amount of memory space in Algorithm~\ref{alg:method} is $O(\sum_{n=1}^{N}|\TR_{n}|)$.
\end{theorem}
\begin{proof}
	In Algorithm~\ref{alg:method}, $\{\{\mass{\TB(a,n)}\}_{a\in\TB_{n}}\}_{n=1}^{N}$, $\{\TR_{n}\}_{n=1}^{N}$, and $\{\TB_{n}\}_{n=1}^{N}$ need to be loaded into memory at once.
	Each has at most $\sum_{n=1}^{N}|\TR_{n}|$ values.
	Thus, the memory requirement is $O(\sum_{n=1}^{N}|\TR_{n}|)$. 
\end{proof}

\subsubsection{Accuracy in \change{Dense-subtensor Detection}}
\label{sec:method:analysis:accuracy}
We show that \method gives the same accuracy guarantee with in-memory algorithms \cite{shin2018fast}, if we set $\theta$ to $1$, although accesses to tuples (stored on disk) are restricted in \method to reduce disk I/Os. 
Specifically, Theorem~\ref{thm:accuracy:guarantee} states that the subtensor found by Algorithm~\ref{alg:method:single} with the maximum cardinality policy has density at least $\frac{1}{\theta N}$ of the optimum when $\densityarinoarg$ is used as the density measure.

\begin{theorem}[$\theta N$-Approximation Guarantee]
	\label{thm:accuracy:guarantee}
	Let $\TB^{*}$ be the subtensor $\TB$ maximizing $\densityari{\TB}{\TR}$ in the given relation $\TR$.
	Let $\tilde{\TB}$ be the subtensor returned by Algorithm~\ref{alg:method:single} with $\densityarinoarg$ and the maximum cardinality policy.
	Then,
	\begin{equation*}
	\densityari{\tilde{\TB}}{\TR} \geq \frac{1}{\theta N}\densityari{\TB^{*}}{\TR}.
	\end{equation*}
\end{theorem}
\begin{proof}

	First, the maximal subtensor $\TB^{*}$ satisfies that, for any $i\in[N]$
	and for any attribute value $a\in\TB_{i}^{*}$, its attribute-value
	mass $\mass{\TB^{*}(a,i)}$ is at least $\frac{1}{N}\densityari{\TB^{*}}{\TR}$.
	This is since the maximality of $\densityari{\TB^{*}}{\TR}$ implies $\densityari{\TB^{*}-\TB^{*}(a,i)}{\TR}\leq\densityari{\TB^{*}}{\TR}$,
	and plugging in Definition \ref{defn:density:ari} to $\rho_{ari}$ gives  $\frac{\mass{\TB^{*}}-\mass{\TB^{*}(a,i)}}{\frac{1}{N}((\sum_{n=1}^{N}|\TB_{n}^{*}|)-1)}=\densityari{\TB^{*}-\TB^{*}(a,i)}{\TR}\leq\densityari{\TB^{*}}{\TR}=\frac{\mass{\TB^{*}}}{\frac{1}{N}\sum_{n=1}^{N}|\TB_{n}^{*}|}$, which reduces to 
	\begin{equation}
	\label{eq:maxblock:density}
	\mass{\TB^{*}(a,i)}\geq\frac{1}{N}\densityari{\TB^{*}}{\TR}.
	\end{equation}
	
	Consider the earliest iteration (lines~\ref{alg:method:single:loop}-\ref{alg:method:line:single:remove}) in Algorithm~\ref{alg:method:single}
	where an attribute value $a$ of $\TB^{*}$
	is included in $D_{i}$. 
	Let $\TB'$ be $\TB$ in the beginning of
	the iteration. Our goal is to prove $\densityari{\tilde{\TB}}{\TR}\geq\frac{1}{\theta N}\densityari{\TB^{*}}{\TR}$,
	which we will show as $\densityari{\tilde{\TB}}{\TR}\geq\densityari{\TB'}{\TR}\geq\frac{\mass{\TB'(a,i)}}{\theta}\geq\frac{\mass{\TB^{*}(a,i)}}{\theta}\geq\frac{1}{\theta N}\densityari{\TB^{*}}{\TR}.$
	
	First, $\densityari{\tilde{\TB}}{\TR}\geq\densityari{\TB'}{\TR}$ is from the maximality of $\densityari{\tilde{\TB}}{\TR}$
	among the densities of the subtensors generated in the iterations (lines \ref{alg:method:line:single:order1}-\ref{alg:method:line:single:order2} in Algorithm~\ref{alg:method:single}). Second,
	applying $|\TB'_{i}|\geq\frac{1}{N}\sum_{n=1}^{N}|\TB'_{n}|$ from
	the maximum cardinality policy (Algorithm \ref{alg:method:guaranteed}) to Definition \ref{defn:density:ari} of $\rho_{ari}$
	gives $\densityari{\TB'}{\TR}=\frac{\mass{\TB'}}{\frac{1}{N}\sum_{n=1}^{N}|\TB'_{n}|}\geq\frac{\mass{\TB'}}{|\TB'_{i}|}$.
	And $a\in D_{i}$ gives $\theta\frac{\mass{\TB'}}{|\TB'_{i}|}\geq\mass{\TB'(a,i)}$.
	So combining these gives $\densityari{\TB'}{\TR}\geq\frac{\mass{\TB'(a,i)}}{\theta}$.
	Third, $\frac{\mass{\TB'(a,i)}}{\theta}\geq\frac{\mass{\TB^{*}(a,i)}}{\theta}$ is from $\TB'\supset \TB^{*}$. Fourth, $\frac{\mass{\TB^{*}(a,i)}}{\theta}\geq\frac{1}{\theta N}\densityari{\TB^{*}}{\TR}$ is from Eq.~\eqref{eq:maxblock:density}. Hence, $\densityari{\tilde{\TB}}{\TR}\geq\frac{1}{\theta N}\densityari{\TB^{*}}{\TR}$ holds. 
\end{proof}

\subsubsection{\change{Theoretical Comparison with \mzoom and \mbiz \cite{shin2018fast}.}}
\change{
	While \method requires only $O(\sum_{n=1}^{N}|\TR_{n}|)$ memory space (see Theorem~\ref{thm:space}), which does not depend on the number of tuples (i.e., $|\TR|$), \mzoom and \mbiz require additional $O(N|\TR|)$ space for storing all tuples in main memory. 
	The worst-case time complexity of \method is $O(kN^2|\TR|\min(\log_{\theta}L, L))$ (see Theorem~\ref{thm:time:worst}), and it is slightly higher than that of \mzoom, which is $O(kN|\TR|\log L)$. Empirically, however, \method is up to $7\times$ faster than \mzoom, as we show in Section~\ref{sec:experiments}.
	The main reason is that \method reads and writes tuples only sequentially, allowing efficient caching based on spatial locality. On the other hand, \mzoom requires tuples to be stored and accessed in hash tables, making efficient caching difficult.\footnote{\change{\mzoom repeats retrieving all tuples with a given attribute value, and thus it requires storing and accessing tuples in hash tables for quick retrievals.}}
	The time complexity of \mbiz depends on the number of iterations until reaching a local optimum, and there is no known upper bound on the number of iterations tighter than $O(2^{(\sum_{n=1}^{N}|\TR_{n}|)})$.
	If $\densityarinoarg$ is used, \mzoom and \mbiz\footnote{\change{We assume that \mbiz uses the outputs of \mzoom as its initial states, as suggested in \cite{shin2018fast}.}} give an approximation ratio of $N$, which is the approximation ratio of \method when $\theta$ is set to $1$ (see Theorem~\ref{thm:accuracy:guarantee}).
}

\subsection{MapReduce Implementation}
\label{sec:method:mr}

We present our \mapreduce implementation of \method, assuming that tuples in relations are stored in a distributed file system.
Specifically, we describe four \mapreduce algorithms that cover the steps of \method accessing tuples.

{\bf (1) Filtering Tuples.}
In lines~\ref{alg:method:line:remove}-\ref{alg:method:line:refine} of Algorithm~\ref{alg:method} and line~\ref{alg:method:line:single:remove} of Algorithm~\ref{alg:method:single}, \method filters the tuples satisfying the given conditions.
These steps are done by the following map-only algorithm, where we broadcast the data used in each condition (e.g., $\{\TB_{n}\}_{n=1}^{N}$ in line~\ref{alg:method:line:remove} of Algorithm~\ref{alg:method}) to mappers using the distributed cache functionality.
\begin{itemize}
	\item Map-stage: Take a tuple $t$ (i.e., $\langle t[A_{1}],..., t[A_{N}]$, $t[X]\rangle$) and emit $t$ if $t$ satisfies the given condition. Otherwise, the tuple is ignored. 
\end{itemize}

{\bf (2) Computing Attribute-value Masses.} Line~\ref{alg:method:line:single:mass} of Algorithm~\ref{alg:method:single} is performed by the following algorithm, where we reduce the amount of shuffled data by combining the intermediate results within each mapper.
\begin{itemize}
	\item Map-stage: Take a tuple $t$ (i.e., $\langle t[A_{1}],..., t[A_{N}]$, $t[X]\rangle$) and emit $N$ key/value pairs $\{\langle (n,t[A_n]), t[X]\rangle \}_{n=1}^{N}$.
	\item Combine-stage/Reduce-stage: Take $\langle (n,a)$, values$\rangle$ and emit $\langle (n, a)$, sum(values)$\rangle$.
\end{itemize}
Each tuple $\langle (n, a)$, value$\rangle$ of the final output indicates that $\mass{\TB(a,n)}=$value.

{\bf (3) Computing Mass.} Line~\ref{alg:method:line:mass} of Algorithm~\ref{alg:method} can be performed by the following algorithm, where we reduce the amount of shuffled data by combining the intermediate results within each mapper.
\begin{itemize}
	\item Map-stage: Take a tuple $t$ (i.e., $\langle t[A_{1}],..., t[A_{N}]$, $t[X]\rangle$) and emit  $\langle 0,t[X]\rangle$.
	\item Combine-stage/Reduce-stage: Take $\langle0$, values$\rangle$ and emit $\langle0$, sum(values)$\rangle$.
\end{itemize}
The value of the final tuple corresponds to $\mass{\TR}$.

{\bf (4) Computing Attribute-value Sets.} Line~\ref{alg:method:line:distinct} of Algorithm~\ref{alg:method} can be performed by the following algorithm, where we reduce the amount of shuffled data by combining the intermediate results within each mapper.
\begin{itemize}
	\item Map-stage: Take a tuple $t$ (i.e., $\langle t[A_{1}],..., t[A_{N}]$, $t[X]\rangle$) and emit $N$ key/value pairs $\{\langle (n,t[A_n]),0\rangle\}_{n=1}^{N}$.
	\item Combine-stage/Reduce-stage: Take $\langle (n,a)$, values$\rangle $ and emit $\langle$($n$, $a$), $0\rangle.$
\end{itemize}
Each tuple $\langle $($n$, $a$), $0\rangle $ of the final output indicates that $a$ is a member of $\TR_{n}$.


\section{Results and Discussion}
\label{sec:experiments}
We designed and conducted experiments to answer the following questions:
\begin{itemize}
	\item {\bf Q1. Memory Efficiency}: How much memory space does \method require for analyzing real-world tensors? How large tensors can \method handle?
	\item {\bf Q2. Speed and Accuracy \change{in Dense-subtensor Detection}}: How rapidly and accurately does \method identify dense subtensors? Does \method outperform its best competitors?
	\item {\bf Q3. Scalability}:  Does \method scale linearly with all aspects of data? Does \method scale out?
	\item {\bf Q4. Effectiveness \change{in Anomaly Detection}}: Which anomalies does \method detect in real-world tensors?  
	\item {\bf Q5. Effect of $\theta$}:  How does the mass-threshold parameter $\theta$ affect the speed and accuracy of \method \change{in dense-subtensor detection}?
	\item {\bf Q6. Effect of $\alpha$}: How does the parameter $\alpha$ in density metric $\densitysurpnoarg$ affect subtensors that \method detects?
\end{itemize}

\subsection{Experimental Settings}
\label{sec:experiments:settings}

\subsubsection{Machines}
We ran all serial algorithms on a machine with 2.67GHz Intel Xeon E7-8837
CPUs and 1TB memory.
We ran \mapreduce algorithms on a 40-node Hadoop cluster, where each node has an Intel Xeon E3-1230 3.3GHz CPU and 32GB memory.

\subsubsection{Datasets}
We describe the real-world and synthetic tensors used in our experiments.
Real-world tensors are categorized into four groups: (a) Rating data (SWM, Yelp, Android, Netflix, and YahooM.), (b) Wikipedia revision histories (KoWiki and EnWiki), (c) Temporal social networks (Youtube and SMS), and (d) TCP dumps (DARPA and AirForce).
Some statistics of these datasets are summarized in Table~\ref{tab:data:real}.

{\bf Rating data.} Rating data are relations with schema (\underline{\smash{user}}, \underline{\smash{item}}, \underline{\smash{timestamp}}, \underline{\smash{score}}, \#ratings). Each tuple ($u$,$i$,$t$,$s$,$r$) indicates that user $u$ gave item $i$ score $s$, $r$ times, at timestamp $t$. In the SWM dataset \cite{akoglu2013opinion}, the timestamps are in dates, and the items are entertaining software from a popular online software marketplace. In the Yelp dataset, the timestamps are in dates, and the items are businesses listed on Yelp, a review site. In the Android dataset \cite{mcauley2015inferring}, the timestamps are hours, and the items are Android apps on Amazon, an online store. In the Netflix dataset \cite{bennett2007netflix}, the timestamps are in dates, and the items are movies listed on Netflix, a movie rental and streaming service. In the YahooM. dataset \cite{dror2012yahoo},  the timestamps are in hours, and the items are musical items listed on Yahoo! Music, a provider of various music services.

{\bf Wikipedia revision history.} Wikipedia revision histories are relations with schema (\underline{\smash{user}}, \underline{\smash{page}}, \underline{\smash{timestamp}}, \#revisions). Each tuple ($u$,$p$,$t$,$r$) indicates that user $u$ revised page $p$, $r$ times, at timestamp $t$ (in hour) in Wikipedia, a crowd-sourcing online encyclopedia. In the KoWiki dataset, the pages are from Korean Wikipedia. In the EnWiki dataset, the pages are from English Wikipedia.

{\bf Temporal social networks.} Temporal social networks are relations with schema (\underline{\smash{source}}, \underline{\smash{destination}}, \underline{\smash{timestamp}}, \#interactions). Each tuple ($s$,$d$,$t$,$i$) indicates that user $s$ interacts with user $d$, $i$ times, at timestamp $t$. In the Youtube dataset \cite{mislove-2007-socialnetworks}, the timestamps are in hours, and the interactions are becoming friends on Youtube, a video-sharing website. In the SMS dataset, the timestamps are in hours, and the interactions are sending text messages.

{\bf TCP Dumps.} The DARPA dataset \cite{lippmann2000evaluating}, collected by the Cyber Systems and Technology Group in 1998, is a relation with schema (\underline{\smash{source IP}}, \underline{\smash{destination IP}}, \underline{\smash{timestamp}}, \#connections). Each tuple ($s$,$d$,$t$,$c$) indicates that $c$ connections were made from IP $s$ to IP $d$ at timestamp $t$ (in minutes).
The AirForce dataset, used for KDD Cup 1999, is a relation with schema (\underline{\smash{protocol}}, \underline{\smash{service}}, \underline{\smash{src bytes}}, \underline{\smash{dst bytes}}, \underline{\smash{flag}}, \underline{\smash{host count}}, \underline{\smash{srv count}}, \#connections). 
The description of each attribute is as follows:
\begin{compactitem}
	\item protocol: type of protocol (tcp, udp, etc.).
	\item service: service on destination (http, telnet, etc.).
	\item src bytes: bytes sent from source to destination.
	\item dst bytes: bytes sent from destination to source.
	\item flag: normal or error status.
	\item host count: number of connections made to the same host in the past two
	seconds.
	\item srv count: number of connections made to the same service in the past two
	seconds.
	\item \#connections: number of connections with the given dimension
	attribute values.
\end{compactitem}


\begin{table}[t]
	\centering
	\caption{\label{tab:data:real} Summary of real-world datasets.} 
	\begin{tabular}{l|l|l}
		\toprule
		{\bf \ \quad Name} & {\bf \qquad\qquad Volume} & {\bf \ \#Tuples} \\
		\midrule
		\multicolumn{3}{l}{{Rating data (\underline{user}, \underline{item}, \underline{\smash{timestamp}}, \underline{\smash{rating}}, \#reviews)}}\\
		\midrule
		SWM \cite{akoglu2013opinion} & 967K $\times$ 15.1K $\times$ 1.38K $\times$ 5 & 1.13M \\
		Yelp & 552K $\times$ 77.1K $\times$ 3.80K $\times$ 5 & 2.23M \\
		Android \cite{mcauley2015inferring} & 1.32M $\times$ 61.3K $\times$ 1.28K $\times$ 5 & 2.64M \\
		Netflix \cite{bennett2007netflix} & 480K $\times$ 17.8K $\times$ 2.18K $\times$ 5 & 99.1M \\
		YahooM. \cite{dror2012yahoo} & 1.00M $\times$ 625K $\times$ 84.4K $\times$ 101 & 253M \\
		\midrule
		\multicolumn{3}{l}{Wiki revision histories (\underline{user}, \underline{\smash{page}}, \underline{\smash{timestamp}}, \#revisions)}\\
		\midrule
		KoWiki & 470K $\times$ 1.18M $\times$ 101K & 11.0M \\
		EnWiki & 44.1M $\times$ 38.5M $\times$  129K & 483M\\
		\midrule
		\multicolumn{3}{l}{{Social networks (\underline{user}, \underline{user}, \underline{\smash{timestamp}}, \#interactions)}}\\
		\midrule
		Youtube \cite{mislove-2007-socialnetworks} & 3.22M $\times$ 3.22M $\times$ 203 & 18.7M \\
		SMS & 1.25M $\times$ 7.00M $\times$ 4.39K & 103M \\
		\midrule
		\multicolumn{3}{l}{{TCP dumps (\underline{\smash{src IP}}, \underline{\smash{dst IP}}, \underline{\smash{timestamp}}, \#connections)}}\\
		\midrule
		DARPA \cite{lippmann2000evaluating} & 9.48K $\times$ 23.4K $\times$ 46.6K & 522K \\
		\midrule
		\multicolumn{3}{l}{{TCP dumps (\underline{\smash{protocol}}, \underline{\smash{service}}, \underline{\smash{src bytes}}, $\cdots$, \#connections) }}\\
		\midrule
		\multirow{2}{*}{AirForce} & 3 $\times$ 70 $\times$ 11 $\times$ 7.20K  & \multirow{2}{*}{648K}  \\
		\ & $\times$ 21.5K $\times$ 512 $\times$ 512 & \\
		\bottomrule
	\end{tabular}
\end{table}

{\bf Synthetic Tensors:} We used synthetic tensors for scalability tests.
Each tensor was created by generating a random binary tensor and injecting ten random dense subtensors, whose volumes are $10^{N}$ and densities (in terms of $\densityarinoarg$) are between 10$\times$ and 100$\times$ of that of the entire tensor.

\begin{figure}[t]

\centering
\includegraphics[width=\linewidth]{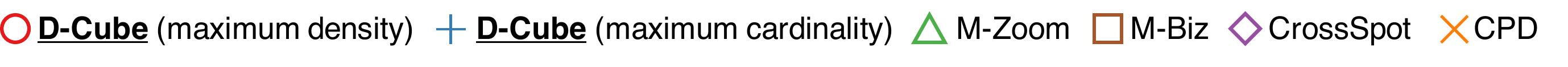} \\
\begin{tabular}{ccccc}
	\begin{minipage}{0.18\textwidth}
		\center
		\includegraphics[width=\linewidth]{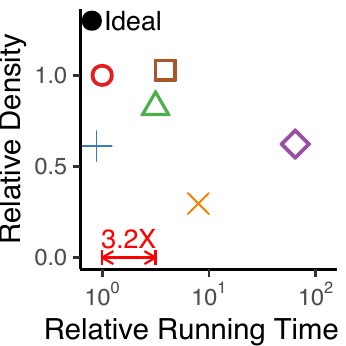}
	\end{minipage} 
	& \begin{minipage}{.18\textwidth}
		\center
		\includegraphics[width=\linewidth]{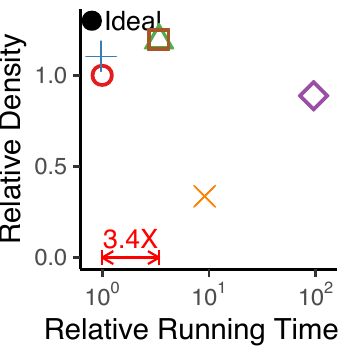}
	\end{minipage} 
	& \begin{minipage}{.18\textwidth}
		\center
		\includegraphics[width=\linewidth]{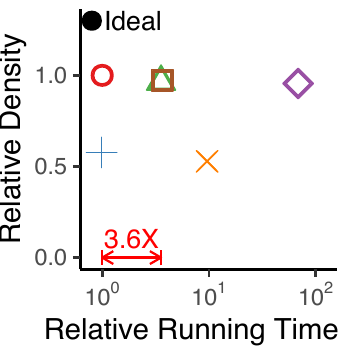}
	\end{minipage} 
	& \begin{minipage}{.18\textwidth}
		\center
		\includegraphics[width=\linewidth]{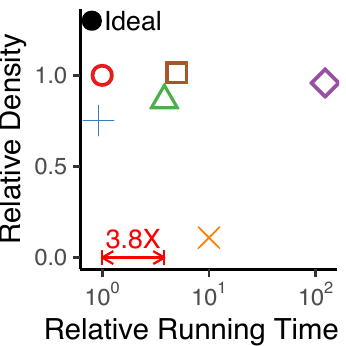}
	\end{minipage} 
	& \begin{minipage}{.18\textwidth}
		\center
		\includegraphics[width=\linewidth]{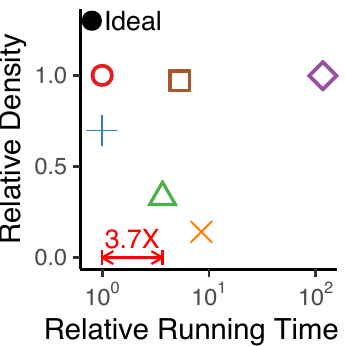}
	\end{minipage} \vspace{1mm} \\
	A. $\densitysuspnoarg$ & B. $\densityarinoarg$ & C. $\densitygeonoarg$ & D. $\densitysurpalpha{1}$ & E. $\densitysurpalpha{10}$
\end{tabular}
	\caption{\label{fig:tradeoff:summary}
		\change{{\bf \method rapidly and accurately detects dense subtensors.} In each plot, points indicate the densities of subtensors detected by different methods and their running times, averaged over all considered real-world tensors.
			Upper-left region indicates better performance. \method is about {\bf 3.6$\times$ faster} than the second fastest method \mzoom. Moreover, \method with the maximum density consistently finds dense subtensors regardless of target density measures.}
	}
\end{figure}

\begin{figure}[t]
	\centering
\includegraphics[width=0.7\linewidth]{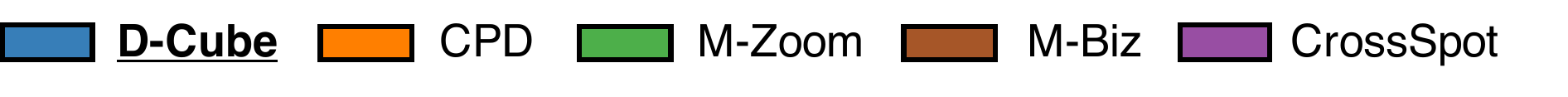} \\
\includegraphics[width=\linewidth]{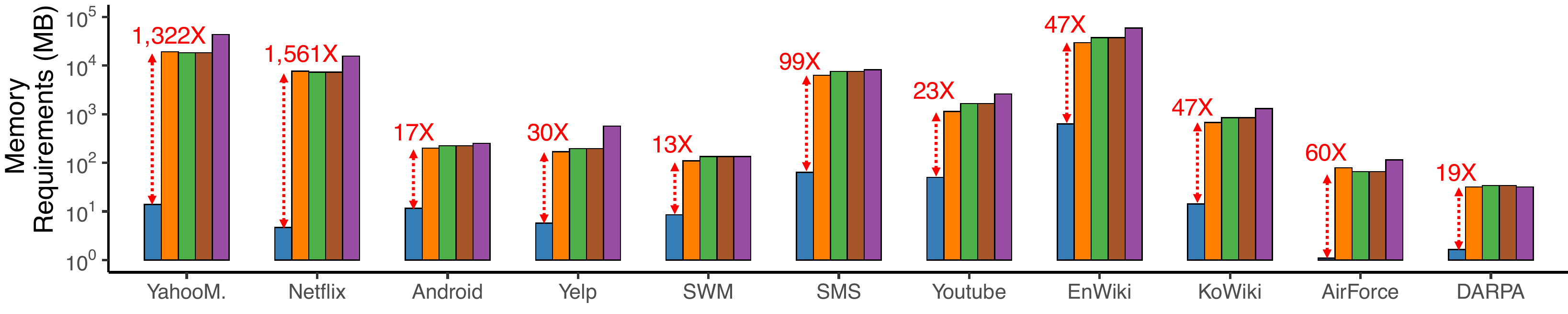}
	\caption{\label{fig:memory}
		{\bf \method is memory efficient.}
		\method requires up to {\bf 1,561$\times$ less memory} than the second most memory-efficient method.
	}
\end{figure}

\subsubsection{Implementations}
We implemented the following dense-subtensor detection methods for our experiments.
\begin{itemize}
	\item \method (Proposed): We implemented \method in Java with Hadoop 1.2.1. We set the mass-threshold parameter $\theta$ to $1$ and used the maximum density policy for dimension selection, unless otherwise stated.
	\item \change{\mzoom \cite{shin2018fast} and \mbiz \cite{shin2018fast}: We used the open-source Java implementations of \mzoom and \mbiz \footnote{https://github.com/kijungs/mzoom}. As suggested in \cite{shin2018fast}, we used the outputs of \mzoom as the initial states in \mbiz.}
	\item \cross \cite{jiang2015general}: We used a Java implementation of the open-source implementation of \cross \footnote{https://github.com/mjiang89/CrossSpot}. Although \cross was originally designed to maximize $\densitysuspnoarg$, we used its variants that directly maximize the density metric compared in each experiment.
	We used CPD as the seed selection method of \cross as in \cite{shin2018fast}.
	\item CPD (CP Decomposition): Let $\{\MA^{(n)}\}_{n=1}^{N}$ be the factor matrices obtained by CP Decomposition \cite{kolda2009tensor}. 
	The $i$-th dense subtensor is composed by every attribute value $a_{n}$ whose corresponding element in the $i$-th column of $A^{(n)}$ is greater than or equal to $1/\sqrt{|\TR_{n}|}$. We used the Tensor Toolbox\footnote{http://www.sandia.gov/~tgkolda/TensorToolbox/} for CP Decomposition.
	\item MAF \cite{maruhashi2011multiaspectforensics}: We used the Tensor Toolbox for CP Decomposition, which MAF is largely based on.
\end{itemize}

\subsection{Q1. Memory Efficiency}
\label{sec:experiments:memory}

We compare the amount of memory required by different methods for handling the real-world datasets.
As seen in Figure~\ref{fig:memory}, \method, which does not require tuples to be stored in memory, needed up to {\bf1,561$\times$ less memory} than the second most memory-efficient method, which stores tuples in memory.

Due to its memory efficiency, \method successfully handled {\bf 1,000$\times$ larger data} than its competitors within a memory budget.
We ran methods on 3-way synthetic tensors with different numbers of tuples (i.e., $|\TR|$), with a memory budget of 16GB per machine.
In every tensor, the cardinality of each dimension attribute was $1/1000$ of the number of tuples, i.e., $|\TR_{n}|=|\TR|/1000$, $\forall n\in[N]$.
Figure~\ref{fig:crown}(a) in Section~\ref{alg:method} shows the result.
The \hadoop implementation of \method successfully spotted dense subtensors in a  tensor with $10^{11}$ tuples ({\bf 2.6TB}), and the serial version of \method successfully spotted dense subtensors in a tensor with $10^{10}$ tuples ({\bf 240GB}), which was the largest tensor that can be stored on a disk.
However, all other methods ran out of memory even on a tensor with $10^{9}$ tuples (21GB).

\subsection{Q2. Speed and Accuracy \change{in Dense-subtensor Detection}}
\label{sec:experiments:speed}

We compare how rapidly and accurately \method (the serial version) and its competitors detect dense subtensors in the real-world datasets.
We measured the wall-clock time (average over three runs) taken for detecting three subtensors by each method, and we measured the maximum density of the three subtensors found by each method using different density measures in Section~\ref{sec:prelim:density}.
For this experiment, we did not limit the memory budget so that every method can handle every dataset.
\method also utilized extra memory space by caching tuples in memory, as explained in Section~\ref{sec:method:alg:impl}.

Figure~\ref{fig:tradeoff:summary} shows the results averaged over all considered datasets.\footnote{
	In each dataset, we measured the relative running time of each method  (compared to the running time of \method with the maximum density policy) and the relative density of detected dense subtensors (compared to the density of subtensors detected by \method with the maximum density policy). Then, we averaged them over all considered datasets.}
The results in each data set can be found in the appendix.
\method provided the best trade-off between speed and accuracy.
Specifically, {\bf \method was up to 7$\times$ faster} (on average 3.6$\times$ faster) than the second fastest method \mzoom.
Moreover, {\bf \method with the maximum density policy spotted high-density subtensors consistently regardless of target density measures}.
\change{Specifically, on average, \method with the maximum density policy was most accurate in dense-subtensor detection when $\densitygeonoarg$ and $\densitysurpalpha{10}$ were used; and it was second most accurate when $\densitysuspnoarg$ and $\densitysurpalpha{1}$ were used.
	When $\densityarinoarg$ was used, \mzoom, \mbiz, and \method with the maximum cardinality policy were on average more accurate than \method with the maximum density policy.
}
Although MAF does not appear in Figures~\ref{fig:tradeoff:summary}, it consistently provided sparser subtensors than CPD with similar speed.

\subsection{Q3. Scalability}
\label{sec:experiments:scalability:data}

We show that \method scales (sub-)linearly with every input factor, i.e., the number of tuples, the number of dimension attributes, and the cardinality of dimension attributes, and the number of subtensors that we aim to find.
To measure the scalability with each factor, we started with finding a dense subtensor in a synthetic tensor with $10^{8}$ tuples and $3$ dimension attributes each of whose cardinality is $10^{5}$.
Then, we measured the running time as we changed one factor at a time while fixing the other factors.
The threshold parameter $\theta$ was fixed to $1$.
As seen in Figure~\ref{fig:scalability}, \method scaled linearly with every factor and sub-linearly with the cardinality of attributes even when $\theta$ was set to its minimum value $1$.
This supports our claim in Section~\ref{sec:method:analysis:complexity} that the worst-case time complexity of \method (Theorem~\ref{thm:time:worst}) is too pessimistic.
This linear scalability of \method held both with enough memory budget (blue solid lines in Figure~\ref{fig:scalability}) to store all tuples and with minimum memory budget (red dashed lines in Figure~\ref{fig:scalability}) to barely meet the requirements
although \method was up to 3$\times$ faster in the former case.

We also evaluate the machine scalability of the \mapreduce implementation of \method. 
We measured its running time taken for finding a dense subtensor in a synthetic tensor with $10^{10}$ tuples and $3$ dimension attributes each of whose cardinality is $10^{7}$, as we increased the number of machines running in parallel from 1 to 40.
Figure~\ref{fig:scalability:machine} shows the changes in the running time and the speed-up, which is defined as $T_{1}/T_{M}$ where $T_{M}$ is the running time with $M$ machines.
The speed-up increased near linearly when a small number of machines were used, while it flattened as
more machines were added due to the overhead in the distributed
system.

\begin{figure}[t!]
	\centering
\includegraphics[width=0.85\linewidth]{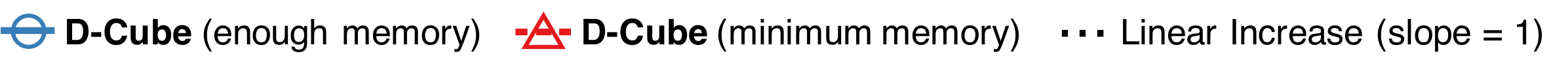} \\
\begin{tabular}{cccc}
	\begin{minipage}{.2\textwidth}
		\center
		\includegraphics[width=\linewidth]{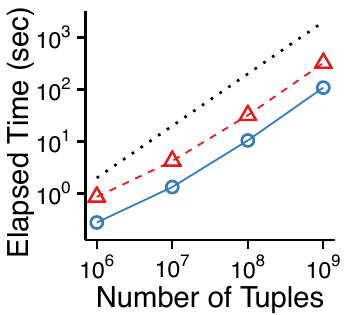}
	\end{minipage} 
	& \begin{minipage}{.2\textwidth}
		\center
		\includegraphics[width=\linewidth]{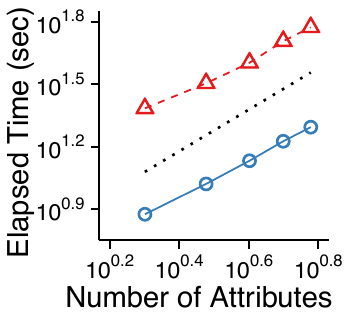}
	\end{minipage} 
	& \begin{minipage}{.2\textwidth}
		\center
		\includegraphics[width=\linewidth]{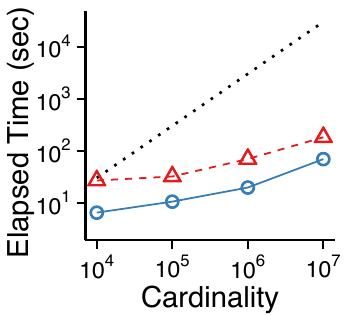}
	\end{minipage} 
	& \begin{minipage}{.2\textwidth}
		\center
		\includegraphics[width=\linewidth]{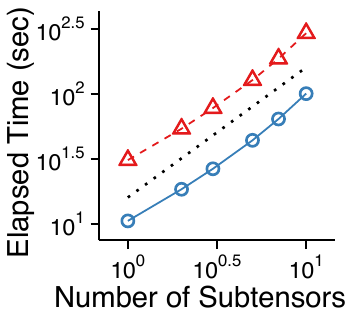}
	\end{minipage}\vspace{1mm} \\
	A. Scalability w.r.t $|\TR|$ & B. Scalability w.r.t $N$ & C. Scalability w.r.t $|\TR_{n}|$ & D.Scalability w.r.t $k$
\end{tabular}
	\caption{\label{fig:scalability}
		{\bf $\method$ scales (sub-)linearly with all input factors} regardless of memory budgets.
	}
\end{figure}

\begin{figure}
	\centering
\begin{tabular}{cc}
	\begin{minipage}{0.5\textwidth}
		\center
		\includegraphics[width=0.4\linewidth]{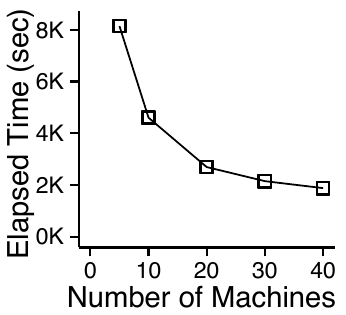}
	\end{minipage} 
	& \begin{minipage}{0.5\textwidth}
		\center
		\includegraphics[width=0.4\linewidth]{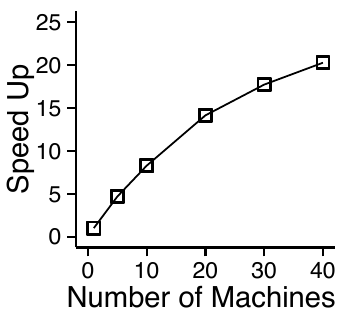}
	\end{minipage}  \vspace{1mm}\\
	A. Elapsed Time & B. Speed Up 
\end{tabular}
	\caption{\label{fig:scalability:machine}
		{\bf $\method$ scales out.}
		The \mapreduce implementation of \method is speeded up 8$\times$ with 10 machines, and 20$\times$ with 40 machines.}
\end{figure}

\subsection{Q4. Effectiveness \change{in Anomaly Detection}}
\label{sec:experiments:effective}

We demonstrate the effectiveness of \method in four applications using real-world tensors.

\subsubsection{Network Intrusion Detection from TCP Dumps} 
\method detected network attacks from TCP dumps accurately by spotting corresponding dense subtensors.
We consider two TCP dumps that are modeled differently.
The DARPA dataset is a 3-way tensor where the dimension attributes are source IPs, destination IPs, and timestamps in minutes; and the measure attribute is the number of connections.
The AirForce dataset, which does not include IP information, is a 7-way tensor where the measure attribute is the same but the dimension attributes are the features of the connections, including protocols and services.
Both datasets include labels indicating whether each connection is malicious or not.

Figure~\ref{fig:crown}(c) in Section~\ref{sec:intro} lists the five densest subtensors (in terms of $\densitygeonoarg$) found by \method in each dataset.
Notice that the dense subtensors are mostly composed of various types of network attacks.
Based on this observation, we classified each connection as malicious or benign based on the density of the densest subtensor including the connection (i.e., the denser the subtensor including a connection is, the more suspicious the connection is).
\change{This led to high area under the ROC curve (AUROC) as seen in Table~\ref{tab:intrusion}, where we report the AUROC when each method was used with the density measure giving the highest AUROC.
	In both datasets, using \method resulted in the highest AUROC.}

\subsubsection{Synchronized Behavior Detection in Rating Data}
\method spotted suspicious synchronized behavior accurately in rating data.
Specifically, we assume an attack scenario where fraudsters in a review site, who aim to boost (or lower) the ratings of the set of items, create multiple user accounts and give the same score to the items within a short period of time.
This lockstep behavior forms a dense subtensor with volume  (\# fake accounts $\times$ \# target items $\times$ 1 $\times$ 1) in the rating dataset, whose dimension attributes are users, items, timestamps, and rating scores.

\begin{table}[t]
	\centering
	\caption{\label{tab:intrusion} \label{tab:review}
		{\bf \method spots network attacks and synchronized behavior fastest and most accurately} from TCP dumps and rating datasets, respectively.
	}
	\scalebox{0.87}{
		\begin{tabular}{c|cc|cc|cc|cc}
			\toprule
			Datasets  & \multicolumn{2}{c|}{AirForce} & \multicolumn{2}{c|}{DARPA}   & \multicolumn{2}{c|}{Android} & \multicolumn{2}{c}{Yelp} \\
			\midrule
			& Elapsed  & \change{AUROC} & Elapsed & \change{AUROC} & Elapsed  & \change{Recall @}& Elapsed & \change{Recall @}\\
			& Time (sec) &  & Time (sec) & & Time (sec) & \change{Top-10} & Time (sec) & \change{Top-10} \\
			\midrule
			CPD \cite{kolda2009tensor} &  413.2 & 0.854 &  105.0 & 0.926 & 59.9 & 0.54 & 47.5 & 0.52 \\
			MAF \cite{maruhashi2011multiaspectforensics}  & 486.6 & 0.912 & 102.4 & 0.514 & 95.0 & 0.54 & 49.4 & 0.52 \\
			\cross \cite{jiang2015general} & 575.5 & 0.924  & 132.2 & 0.923  & 71.3 & 0.54  & 56.7 & 0.52  \\
			\mzoom \cite{shin2018fast} & 27.7 & 0.975 & 22.7 &  0.923 & 28.4 & 0.70 & 17.7 & 0.30 \\
			\change{\mbiz \cite{shin2018fast}} & \change{29.8} & \change{0.977} & \change{22.7}  & \change{0.923} & \change{30.6} & \change{0.70} & \change{19.5} & \change{0.30} \\
			\midrule
			{\bf \method} & {\bf 15.6} & {\bf 0.987} &  {\bf 9.1} & {\bf 0.930} & {\bf 7.0} & {\bf 0.90} & {\bf 4.9} &  {\bf 0.60} \\
			\bottomrule
		\end{tabular}
	}
\end{table}


\begin{table}[t]
	\centering
	\caption{\label{tab:spam} \textbf{\method successfully detects spam reviews} in the SWM dataset.}
	\begin{tabular}{clc|clc}
		\toprule
		\multicolumn{3}{c|}{Subtensor 1 (100\% spam)} & \multicolumn{3}{c}{Subtensor 2 (100\% spam)}  \\
		\midrule
		User & Review &  Date & User & Review  & Date \\
		\midrule
		Ti* & {\scriptsize type in *** and you will get ... } & Mar-4 & Sk* & {\scriptsize invite code***, referral  ...} & Apr-18  \\
		Fo* & {\scriptsize type in for the bonus code: ...} & Mar-4 & fu* & {\scriptsize use my code for bonus ...} & Apr-18 \\
		dj* & {\scriptsize typed in the code: ***  ...}  & Mar-4 & Ta* & {\scriptsize  enter the code *** for ...} & Apr-18 \\
		Di* & {\scriptsize enter this code to start with ...}  & Mar-4 & Ap* & {\scriptsize bonus code *** for points ...} & Apr-18 \\
		Fe* & {\scriptsize enter code: *** to win even ...}  & Mar-4 & De* & {\scriptsize bonus code: ***, be one ...} & Apr-18 \\
		\bottomrule
	\end{tabular}
	\begin{tabular}{clc}
		\toprule
		\multicolumn{3}{c}{Subtensor 3 (at least 48\% spam)} \\
		\midrule
		User & Review & Date \\
		\midrule
		Mr* & {\scriptsize entered this code and got ... } & Nov-23 \\
		Max* & {\scriptsize enter the bonus code: *** ... } & Nov-23\\
		Je* & {\scriptsize enter *** when it asks... } & Nov-23\\
		Man* & {\scriptsize just enter *** for a boost ... } & Nov-23\\
		Ty* & {\scriptsize enter *** ro receive a ...} & Nov-23\\
		\bottomrule
	\end{tabular}
\end{table}

We injected 10 such random dense subtensors whose volumes varied from 15$\times$15$\times$1$\times$1 to 60$\times$60$\times$1$\times$1 in the Yelp and Android datasets.
We compared the ratio of the injected subtensors detected by each dense-subtensor detection method.
We considered each injected subtensor as overlooked by a method if the subtensor did not belong to any of the top-$10$ dense subtensors spotted by the method or it was hidden in a natural dense subtensor at least 10 times larger than the injected subtensor.
\change{That is, we measured the recall at top $10$.}
We repeated this experiment 10 times, and the averaged results are summarized in Table~\ref{tab:review}.
For each method, we report the results with the density measure giving the highest \change{recall}.
In both datasets, \method detected a largest number of the injected subtensors.
Especially, in the Android dataset, \method detected 9 out of the 10 injected subtensors, while the second best method detected only 7 injected subtensors on average. 

\subsubsection{Spam-Review Detection in Rating Data}
\method successfully spotted spam reviews in the SWM dataset, which contains reviews from an online software marketplace. We modeled the SWM dataset as a 4-way tensor whose dimension attributes are users, software, ratings, and timestamps in dates, and we applied \method (with $\density=\densityarinoarg$) to the dataset.
Table~\ref{tab:blocks:summary} shows the statistics of the top-$3$ dense subtensors.
Although ground-truth labels were not available, as the examples in Table~\ref{tab:spam} show, all the reviews composing the first and second dense subtensors were obvious spam reviews. In addition, at least $48\%$ of the reviews composing the third dense subtensor were obvious spam reviews.

\subsubsection{Anomaly Detection in Wikipedia Revision Histories}
\method detected interesting anomalies in Wikipedia revision histories, which we model as 3-way tensors whose dimension attributes are users, pages, and timestamps in hours.
Table~\ref{tab:blocks:summary} gives the statistics of the top-$3$ dense subtensors detected by \method (with $\density=\densityarinoarg$ and the maximum cardinality policy) in the KoWiki dataset and by \method (with $\density=\densitygeonoarg$ and the maximum density policy) in the EnWiki dataset.
All three subtensors detected in the KoWiki dataset indicated edit wars. 
For example, the second subtensor corresponded to an edit war where 4 users changed 4 pages, 1,011 times, within 5 hours. 
On the other hand, all three subtensors detected in the Enwiki dataset indicated bot activities. For example, the third subtensor corresponded to 3 bots which edited 1,067 pages 973,747 times. The users composing the top-$5$ dense subtensors in the EnWiki dataset are listed in Table~\ref{tab:bot}. Notice that all of them are bots.

\begin{table}[t]
	\centering
	\caption{\label{tab:blocks:summary}Summary of the dense subtensors that \method detects in the SWM, KoWiki, and EnWiki datasets.}
	\begin{tabular}{cc|cccc}
		\toprule
		Dataset & Order & Volume & Mass & $\densityarinoarg$ & Type \\
		\midrule
		\multirow{3}{*}{SWM}& 1 & 120 & 308  & 44.0 & Spam reviews \\
		& 2 & 612 & 435 & 31.6 & Spam reviews\\
		& 3 & 231,240 & 771 & 20.3 & Spam reviews\\
		\midrule
		\multirow{3}{*}{KoWiki}& 1 & 8 & 546  & 273.0 & Edit war\\
		& 2 & 80 & 1,011 & 233.3 & Edit war\\
		& 3 & 270 & 1,126 & 168.9 & Edit war\\
		\midrule
		\multirow{3}{*}{EnWiki} & 1 & 9.98M & 1.71M &  7,931 & Bot activities \\
		& 2 & 541K &  343K & 4,211 &  Bot activities\\
		& 3 & 23.5M & 973K & 3,395 & Bot activities\\
		\bottomrule
	\end{tabular}
\end{table}


\begin{table}[t]
	\centering
	\caption{\label{tab:bot} \textbf{\method successfully spots bot activities} in the EnWiki dataset.}
	\begin{tabular}{cl}
		\toprule
		Subtensor \# & Users in each subtensor (100\% bots) \\ 
		\midrule
		1 & WP 1.0 bot \\
		2 & AAlertBot \\
		3 & AlexNewArtBot, VeblenBot, InceptionBot\\
		4 & WP 1.0 bot \\
		5 & Cydebot, VeblenBot \\ 
		\bottomrule
	\end{tabular}
\end{table}

\begin{figure}[t]
	\centering
\includegraphics[width=0.5\linewidth]{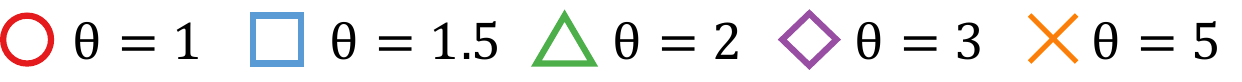} \\
\begin{tabular}{ccccc}
	\begin{minipage}{.18\textwidth}
		\center
		\includegraphics[width=\linewidth]{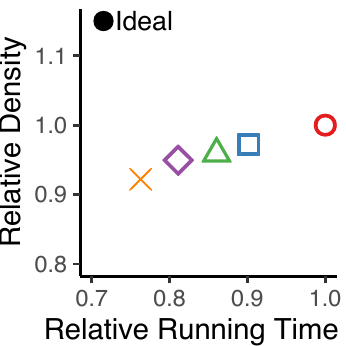}
	\end{minipage} 
	& \begin{minipage}{.18\textwidth}
		\center
		\includegraphics[width=\linewidth]{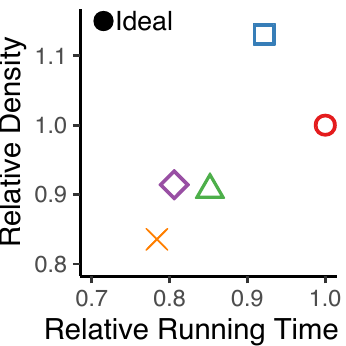}
	\end{minipage} 
	& \begin{minipage}{.18\textwidth}
		\center
		\includegraphics[width=\linewidth]{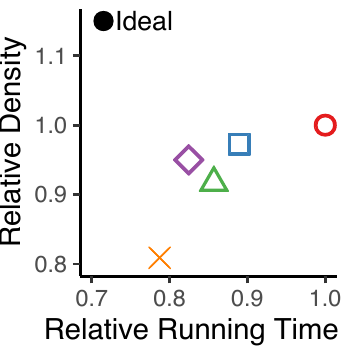}
	\end{minipage} 
	& \begin{minipage}{.18\textwidth}
		\center
		\includegraphics[width=\linewidth]{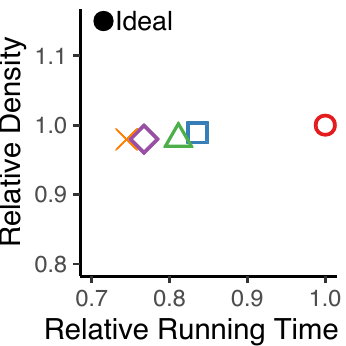}
	\end{minipage} 
	& \begin{minipage}{.18\textwidth}
		\center
		\includegraphics[width=\linewidth]{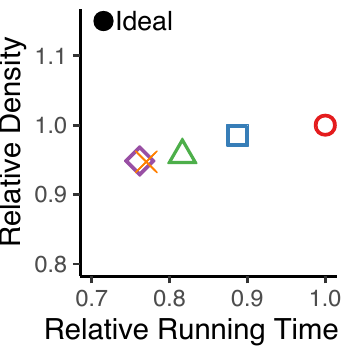}
	\end{minipage} \\
	A. $\densitysuspnoarg$ & B. $\densityarinoarg$ & C. $\densitygeonoarg$ & D. $\densitysurpalpha{1}$ & E. $\densitysurpalpha{10}$
\end{tabular}
	\caption{\label{fig:tradeoff:epsilon}
		{\bf The mass-threshold parameter $\theta$ gives a trade-off between the speed and accuracy of \method in \change{dense-subtensor detection}}. We report the running time and the density of detected subtensors, averaged over all considered real-world datasets. As $\theta$ increases, \method tends to be faster, detecting sparser subtensors.
	}
\end{figure}

\subsection{Q5. Effects of Parameter $\theta$ on Speed and Accuracy \change{in Dense-subtensor Detection}}
\label{sec:experiments:epsilon}

We investigate the effects of the mass-threshold parameter $\theta$ on the speed and accuracy of \method in \change{dense-subtensor detection}. We used the serial version of \method with a memory budget of 16GB, and we measured the relative density of detected subtensors and its running time, as in Section~\ref{sec:experiments:speed}. 
Figure~\ref{fig:tradeoff:epsilon} shows the results averaged over all considered datasets. Different $\theta$ values provided a trade-off between speed and accuracy in dense-subtensor detection. Specifically, increasing $\theta$ tended to make \method faster but also make it detect sparser subtensors. This tendency is consistent with our theoretical analyses (Theorems~\ref{thm:time:worst}-\ref{thm:accuracy:guarantee} in Section~\ref{sec:method:analysis}). The sensitivity of the \change{dense-subtensor detection} accuracy to $\theta$ depended on the used density measures. Specifically, the sensitivity was lower with $\densitysurpnoarg$ than with the other density measures. \ \\

\subsection{Q6. Effects of Parameter $\alpha$ in $\densitysurpnoarg$ on Subtensors Detected by \method}
\label{sec:exp:surp}
We show that the dense subtensors detected by \method are configurable by the parameter $\alpha$ in density measure $\densitysurpnoarg$. Figure~\ref{fig:tradeoff:alpha} shows the volumes and masses of subtensors detected in the Youtube and Yelp datasets by \method when $\densitysurpnoarg$ with different $\alpha$ values were used as the density metrics.
With large $\alpha$ values, \method tended to spot relatively small but compact subtensors.
With small $\alpha$ values, however, \method tended to spot relatively sparse but large subtensors.
Similar tendencies were obtained with the other datasets.

\begin{figure}[t]
	\centering
\includegraphics[width=0.48\linewidth]{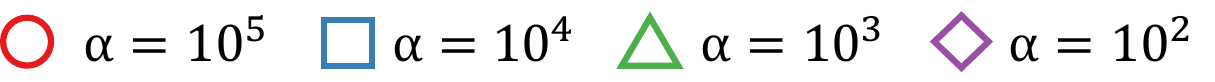} \\
\begin{tabular}{cc}
	\begin{minipage}{.5\textwidth}
		\center
		\includegraphics[width=0.4\linewidth]{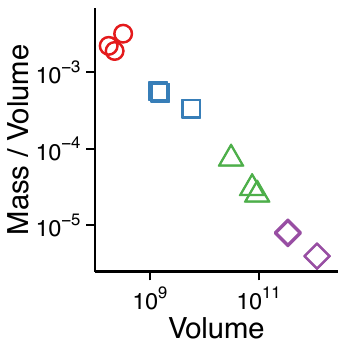}
	\end{minipage} 
	& \begin{minipage}{.5\textwidth}
		\center
		\includegraphics[width=0.4\linewidth]{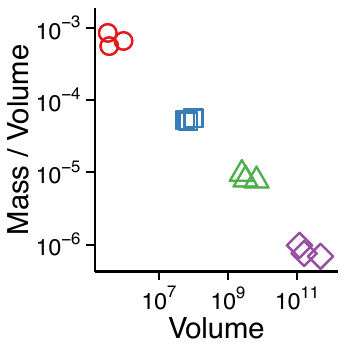}
	\end{minipage} \vspace{1mm}  \\
	A. Youtube & B. Yelp
\end{tabular}
	\caption{\label{fig:tradeoff:alpha}
		{\bf Subtensors detected by \method are configurable} by the parameter $\alpha$ in density metric $\densitysurpnoarg$. As $\alpha$ increases, \method spots smaller but more compact subtensors. 
	}
\end{figure}

\section{Conclusion}
\label{sec:conclusion}
In this work, we propose \method, a disk-based dense-subtensor detection method, to deal with disk-resident tensors too large to fit in main memory.
\method is optimized to minimize disk I/Os while providing a guarantee on the quality of the subtensors it finds.
Moreover, we propose a distributed version of \method running on \mapreduce for terabyte-scale or larger data distributed across multiple machines.
In summary, \method achieves the following advantages over its state-of-the-art competitors:
\begin{itemize}
	\item {\bf Memory Efficient:} 
	\method handles {\it 1,000$\times$} larger data ({\it 2.6TB}) by reducing memory usage up to {\it 1,561$\times$} compared to in-memory algorithms (Section~\ref{sec:experiments:memory}).
	\item {\bf Fast:} Even when data fit in memory, \method is up to {\it 7$\times$} faster than its competitors (Section~\ref{sec:experiments:speed}) with near-linear scalability (Section~\ref{sec:experiments:scalability:data}).
	\item {\bf Provably Accurate:}  \method is one of the methods guaranteeing the best approximation ratio (Theorem~\ref{thm:accuracy:guarantee}) \change{in dense-subtensor detection} and spotting the densest subtensors in practice (Section~\ref{sec:experiments:speed}). 
	\item {\bf Effective:} \method was most accurate in two applications:  detecting network attacks from TCP dumps and lockstep behavior in rating data (Section~\ref{sec:experiments:effective}).
\end{itemize}
\noindent{\bf Reproducibility:} The code and data used in the paper are available at {\bf \url{http://dmlab.kaist.ac.kr/dcube}}.

%

\section*{Acknowledgments}
This research was supported by National Research Foundation of Korea (NRF) grant funded by the Korea government (MSIT) (No. NRF-2020R1C1C1008296) and Institute of Information \& Communications Technology Planning \&
Evaluation (IITP) grant funded by the Korea government (MSIT) (No. 2019-0-00075, Artificial Intelligence Graduate School Program (KAIST)).
This research was also supported by the National Science Foundation under Grant No. CNS-1314632 and IIS-1408924. 
This research was sponsored by the Army Research Laboratory and was accomplished under Cooperative Agreement Number W911NF-09-2-0053. 
Any opinions, findings, and conclusions or recommendations expressed in this material are those of the author(s) and do not necessarily reflect the views of the National Science Foundation, or other funding parties. The U.S. Government is authorized to reproduce and distribute reprints for Government purposes notwithstanding any copyright notation here on.

\bibliographystyle{plain}
\bibliography{kijung}

\begin{thebibliography}{10}

\bibitem{akoglu2013opinion}
Leman Akoglu, Rishi Chandy, and Christos Faloutsos.
\newblock Opinion fraud detection in online reviews by network effects.
\newblock {\em ICWSM}, 2013.

\bibitem{akoglu2010oddball}
Leman Akoglu, Mary McGlohon, and Christos Faloutsos.
\newblock Oddball: Spotting anomalies in weighted graphs.
\newblock In {\em PAKDD}. 2010.

\bibitem{akoglu2015graph}
Leman Akoglu, Hanghang Tong, and Danai Koutra.
\newblock Graph based anomaly detection and description: a survey.
\newblock {\em Data Mining and Knowledge Discovery}, 29(3):626--688, 2015.

\bibitem{andersen2009finding}
Reid Andersen and Kumar Chellapilla.
\newblock Finding dense subgraphs with size bounds.
\newblock In {\em WAW}. 2009.

\bibitem{bahmani2014efficient}
Bahman Bahmani, Ashish Goel, and Kamesh Munagala.
\newblock Efficient primal-dual graph algorithms for mapreduce.
\newblock In {\em WAW}, 2014.

\bibitem{bahmani2012densest}
Bahman Bahmani, Ravi Kumar, and Sergei Vassilvitskii.
\newblock Densest subgraph in streaming and mapreduce.
\newblock {\em PVLDB}, 5(5):454--465, 2012.

\bibitem{balalau2015finding}
Oana~Denisa Balalau, Francesco Bonchi, TH~Chan, Francesco Gullo, and Mauro
  Sozio.
\newblock Finding subgraphs with maximum total density and limited overlap.
\newblock In {\em WSDM}, 2015.

\bibitem{bennett2007netflix}
James Bennett and Stan Lanning.
\newblock The netflix prize.
\newblock In {\em KDD Cup}, 2007.

\bibitem{beutel2013copycatch}
Alex Beutel, Wanhong Xu, Venkatesan Guruswami, Christopher Palow, and Christos
  Faloutsos.
\newblock Copycatch: stopping group attacks by spotting lockstep behavior in
  social networks.
\newblock In {\em WWW}, 2013.

\bibitem{charikar2000greedy}
Moses Charikar.
\newblock Greedy approximation algorithms for finding dense components in a
  graph.
\newblock In {\em APPROX}. 2000.

\bibitem{dean2008mapreduce}
Jeffrey Dean and Sanjay Ghemawat.
\newblock Mapreduce: simplified data processing on large clusters.
\newblock {\em Communications of the ACM}, 51(1):107--113, 2008.

\bibitem{dror2012yahoo}
Gideon Dror, Noam Koenigstein, Yehuda Koren, and Markus Weimer.
\newblock The yahoo! music dataset and kdd-cup'11.
\newblock In {\em KDD Cup}, 2012.

\bibitem{epasto2015efficient}
Alessandro Epasto, Silvio Lattanzi, and Mauro Sozio.
\newblock Efficient densest subgraph computation in evolving graphs.
\newblock In {\em WWW}, 2015.

\bibitem{galbrun2016top}
Esther Galbrun, Aristides Gionis, and Nikolaj Tatti.
\newblock Top-k overlapping densest subgraphs.
\newblock {\em Data Mining and Knowledge Discovery}, 30(5):1134--1165, 2016.

\bibitem{goldberg1984finding}
Andrew~V Goldberg.
\newblock {\em Finding a maximum density subgraph}.
\newblock Technical Report, 1984.

\bibitem{hooi2017graph}
Bryan Hooi, Kijung Shin, Hyun~Ah Song, Alex Beutel, Neil Shah, and Christos
  Faloutsos.
\newblock Graph-based fraud detection in the face of camouflage.
\newblock {\em ACM Transactions on Knowledge Discovery from Data}, 11(4):44,
  2017.

\bibitem{jeon2015haten2}
Inah Jeon, Evangelos~E Papalexakis, U~Kang, and Christos Faloutsos.
\newblock Haten2: Billion-scale tensor decompositions.
\newblock In {\em ICDE}, pages 1047--1058, 2015.

\bibitem{jiang2015general}
Meng Jiang, Alex Beutel, Peng Cui, Bryan Hooi, Shiqiang Yang, and Christos
  Faloutsos.
\newblock A general suspiciousness metric for dense blocks in multimodal data.
\newblock In {\em ICDM}, 2015.

\bibitem{jiang2014catchsync}
Meng Jiang, Peng Cui, Alex Beutel, Christos Faloutsos, and Shiqiang Yang.
\newblock Catchsync: catching synchronized behavior in large directed graphs.
\newblock In {\em KDD}, 2014.

\bibitem{kang2012gigatensor}
U~Kang, Evangelos Papalexakis, Abhay Harpale, and Christos Faloutsos.
\newblock Gigatensor: scaling tensor analysis up by 100 times-algorithms and
  discoveries.
\newblock In {\em KDD}, pages 316--324, 2012.

\bibitem{kannan1999analyzing}
Ravi Kannan and V~Vinay.
\newblock {\em Analyzing the structure of large graphs}.
\newblock Technical Report, 1999.

\bibitem{khuller2009finding}
Samir Khuller and Barna Saha.
\newblock On finding dense subgraphs.
\newblock In {\em ICALP}, pages 597--608. 2009.

\bibitem{kolda2009tensor}
Tamara~G Kolda and Brett~W Bader.
\newblock Tensor decompositions and applications.
\newblock {\em SIAM review}, 51(3):455--500, 2009.

\bibitem{lee2010survey}
Victor~E Lee, Ning Ruan, Ruoming Jin, and Charu Aggarwal.
\newblock A survey of algorithms for dense subgraph discovery.
\newblock In {\em Managing and Mining Graph Data}, pages 303--336. 2010.

\bibitem{lippmann2000evaluating}
Richard~P Lippmann, David~J Fried, Isaac Graf, Joshua~W Haines, Kristopher~R
  Kendall, David McClung, Dan Weber, Seth~E Webster, Dan Wyschogrod, Robert~K
  Cunningham, et~al.
\newblock Evaluating intrusion detection systems: The 1998 darpa off-line
  intrusion detection evaluation.
\newblock In {\em DISCEX}, 2000.

\bibitem{maruhashi2011multiaspectforensics}
Koji Maruhashi, Fan Guo, and Christos Faloutsos.
\newblock Multiaspectforensics: Pattern mining on large-scale heterogeneous
  networks with tensor analysis.
\newblock In {\em ASONAM}, 2011.

\bibitem{mcauley2015inferring}
Julian McAuley, Rahul Pandey, and Jure Leskovec.
\newblock Inferring networks of substitutable and complementary products.
\newblock In {\em KDD}, 2015.

\bibitem{mislove-2007-socialnetworks}
Alan Mislove, Massimiliano Marcon, Krishna~P. Gummadi, Peter Druschel, and
  Bobby Bhattacharjee.
\newblock {Measurement and Analysis of Online Social Networks}.
\newblock In {\em IMC}, 2007.

\bibitem{oh2017shot}
Jinoh Oh, Kijung Shin, Evangelos~E. Papalexakis, Christos Faloutsos, and Hwanjo
  Yu.
\newblock S-hot: Scalable high-order tucker decomposition.
\newblock In {\em WSDM}, 2017.

\bibitem{papalexakis2012parcube}
Evangelos~E Papalexakis, Christos Faloutsos, and Nicholas~D Sidiropoulos.
\newblock Parcube: Sparse parallelizable tensor decompositions.
\newblock In {\em PKDD}, 2012.

\bibitem{rossi2013modeling}
Ryan~A Rossi, Brian Gallagher, Jennifer Neville, and Keith Henderson.
\newblock Modeling dynamic behavior in large evolving graphs.
\newblock In {\em WSDM}, 2013.

\bibitem{ruhl2003efficient}
Jan~Matthias Ruhl.
\newblock {\em Efficient algorithms for new computational models}.
\newblock PhD thesis, Massachusetts Institute of Technology, 2003.

\bibitem{saha2010dense}
Barna Saha, Allison Hoch, Samir Khuller, Louiqa Raschid, and Xiao-Ning Zhang.
\newblock Dense subgraphs with restrictions and applications to gene annotation
  graphs.
\newblock In {\em RECOMB}, 2010.

\bibitem{shah2014spotting}
Neil Shah, Alex Beutel, Brian Gallagher, and Christos Faloutsos.
\newblock Spotting suspicious link behavior with fbox: An adversarial
  perspective.
\newblock In {\em ICDM}, 2014.

\bibitem{shin2016corescope}
Kijung Shin, Tina Eliassi-Rad, and Christos Faloutsos.
\newblock Corescope: Graph mining using k-core analysis - patterns, anomalies
  and algorithms.
\newblock In {\em ICDM}, 2016.

\bibitem{shin2018fast}
Kijung Shin, Bryan Hooi, and Christos Faloutsos.
\newblock Fast, accurate, and flexible algorithms for dense subtensor mining.
\newblock {\em ACM Transactions on Knowledge Discovery from Data},
  12(3):28:1--28:30, 2018.

\bibitem{shin2017dcube}
Kijung Shin, Bryan Hooi, Jisu Kim, and Christos Faloutsos.
\newblock D-cube: Dense-block detection in terabyte-scale tensors.
\newblock In {\em WSDM}, 2017.

\bibitem{shin2017densealert}
Kijung Shin, Bryan Hooi, Jisu Kim, and Christos Faloutsos.
\newblock Densealert: Incremental dense-subtensor detection in tensor streams.
\newblock In {\em Proceedings of the 23rd ACM SIGKDD International Conference
  on Knowledge Discovery and Data Mining}, pages 1057--1066. ACM, 2017.

\bibitem{shin2014distributed}
Kijung Shin and U~Kang.
\newblock Distributed methods for high-dimensional and large-scale tensor
  factorization.
\newblock In {\em ICDM}, 2014.

\bibitem{tsourakakis2013denser}
Charalampos Tsourakakis, Francesco Bonchi, Aristides Gionis, Francesco Gullo,
  and Maria Tsiarli.
\newblock Denser than the densest subgraph: extracting optimal quasi-cliques
  with quality guarantees.
\newblock In {\em KDD}, 2013.

\bibitem{wang2015fast}
Yining Wang, Hsiao-Yu Tung, Alex~J Smola, and Anima Anandkumar.
\newblock Fast and guaranteed tensor decomposition via sketching.
\newblock In {\em NIPS}, 2015.

\end{thebibliography}

\newpage 

\section*{Appendix: Additional Figures}

Figures~\ref{fig:tradeoff:lac_sms}-\ref{fig:tradeoff:yahoo} show the speed and accuracy of the considered algorithms in $11$ different datasets.

\vspace{10mm}

\begin{figure}[h]
	\centering
	\includegraphics[width=\linewidth]{FIG/tradeoff_label.pdf} \\
	\begin{tabular}{ccccc}
		\begin{minipage}{0.18\textwidth}
			\center
			\includegraphics[width=\linewidth]{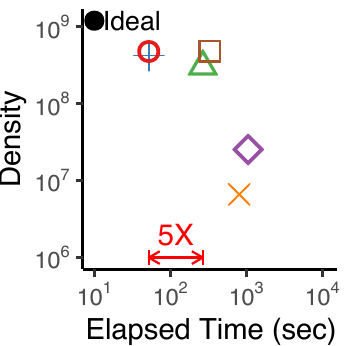}
		\end{minipage} 
		& \begin{minipage}{.18\textwidth}
			\center
			\includegraphics[width=\linewidth]{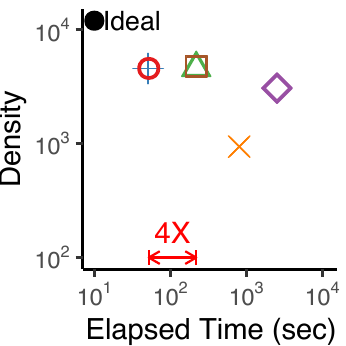}
		\end{minipage} 
		& \begin{minipage}{.18\textwidth}
			\center
			\includegraphics[width=\linewidth]{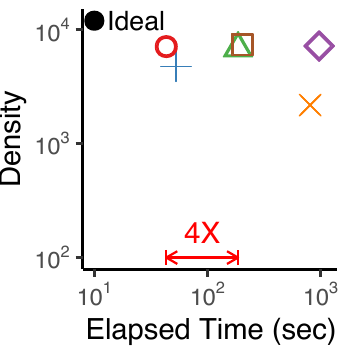}
		\end{minipage} 
		& \begin{minipage}{.18\textwidth}
			\center
			\includegraphics[width=\linewidth]{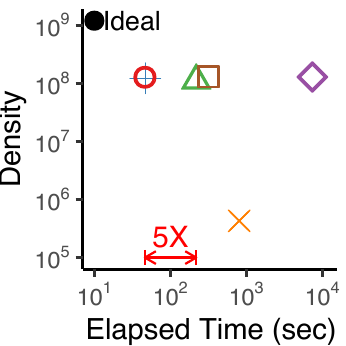}
		\end{minipage} 
		& \begin{minipage}{.18\textwidth}
			\center
			\includegraphics[width=\linewidth]{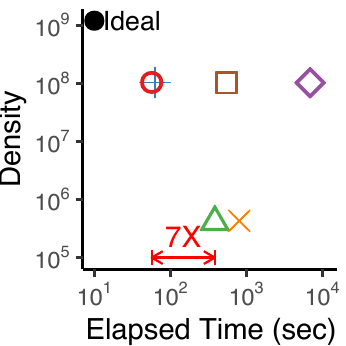}
		\end{minipage} \vspace{1mm} \\
		A. $\densitysuspnoarg$ & B. $\densityarinoarg$ & C. $\densitygeonoarg$ & D. $\densitysurpalpha{1}$ & E. $\densitysurpalpha{10}$
	\end{tabular}
	\caption{\label{fig:tradeoff:lac_sms}
		Speed and accuracy of the algorithms in the SMS dataset.
	}
\end{figure}

\begin{figure}[h]
	\centering
	\includegraphics[width=\linewidth]{FIG/tradeoff_label.pdf} \\
	\begin{tabular}{ccccc}
		\begin{minipage}{0.18\textwidth}
			\center
			\includegraphics[width=\linewidth]{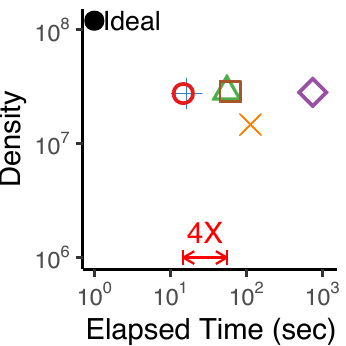}
		\end{minipage} 
		& \begin{minipage}{.18\textwidth}
			\center
			\includegraphics[width=\linewidth]{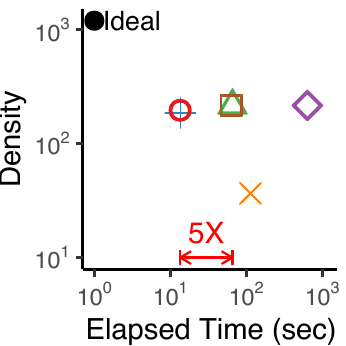}
		\end{minipage} 
		& \begin{minipage}{.18\textwidth}
			\center
			\includegraphics[width=\linewidth]{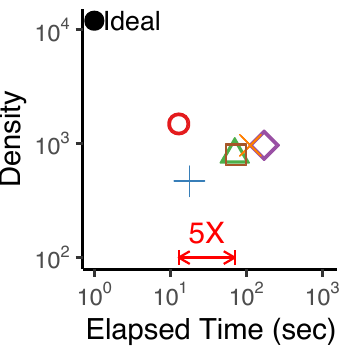}
		\end{minipage} 
		& \begin{minipage}{.18\textwidth}
			\center
			\includegraphics[width=\linewidth]{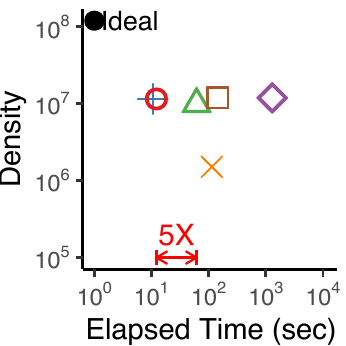}
		\end{minipage} 
		& \begin{minipage}{.18\textwidth}
			\center
			\includegraphics[width=\linewidth]{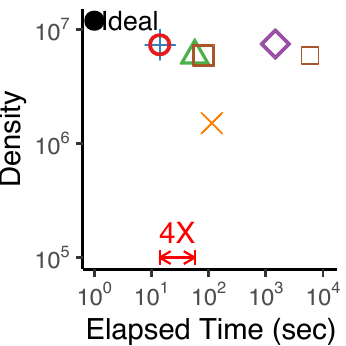}
		\end{minipage} \vspace{1mm} \\
		A. $\densitysuspnoarg$ & B. $\densityarinoarg$ & C. $\densitygeonoarg$ & D. $\densitysurpalpha{1}$ & E. $\densitysurpalpha{10}$
	\end{tabular}
	\caption{\label{fig:tradeoff:youtube}
		Speed and accuracy of the algorithms in the Youtube dataset.
	}
\end{figure}

\begin{figure}[h]
	\centering
	\includegraphics[width=\linewidth]{FIG/tradeoff_label.pdf} \\
	\begin{tabular}{ccccc}
		\begin{minipage}{0.18\textwidth}
			\center
			\includegraphics[width=\linewidth]{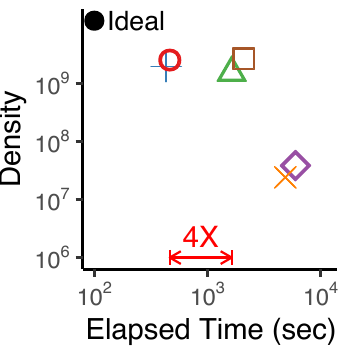}
		\end{minipage} 
		& \begin{minipage}{.18\textwidth}
			\center
			\includegraphics[width=\linewidth]{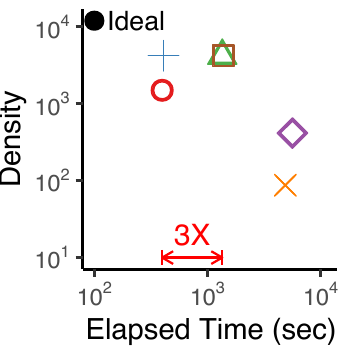}
		\end{minipage} 
		& \begin{minipage}{.18\textwidth}
			\center
			\includegraphics[width=\linewidth]{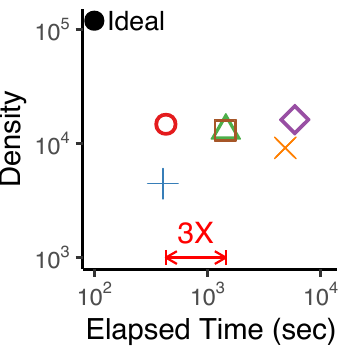}
		\end{minipage} 
		& \begin{minipage}{.18\textwidth}
			\center
			\includegraphics[width=\linewidth]{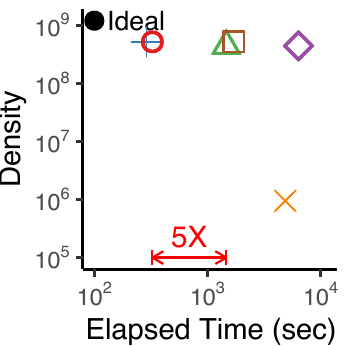}
		\end{minipage} 
		& \begin{minipage}{.18\textwidth}
			\center
			\includegraphics[width=\linewidth]{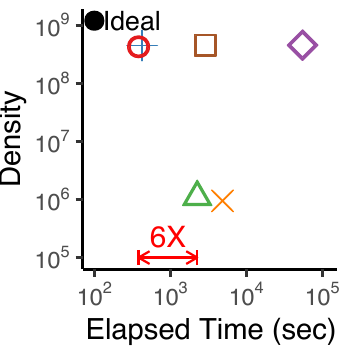}
		\end{minipage} \vspace{1mm} \\
		A. $\densitysuspnoarg$ & B. $\densityarinoarg$ & C. $\densitygeonoarg$ & D. $\densitysurpalpha{1}$ & E. $\densitysurpalpha{10}$
	\end{tabular}
	\caption{\label{fig:tradeoff:enwiki}
		Speed and accuracy of the algorithms in the EnWiki dataset.
	}
\end{figure}

\begin{figure}[h]
	\centering
	\includegraphics[width=\linewidth]{FIG/tradeoff_label.pdf} \\
	\begin{tabular}{ccccc}
		\begin{minipage}{0.18\textwidth}
			\center
			\includegraphics[width=\linewidth]{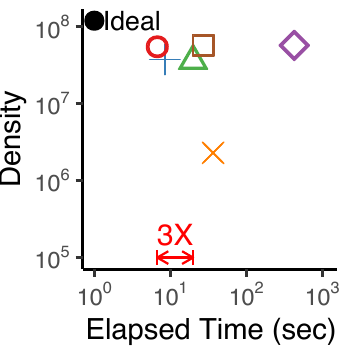}
		\end{minipage} 
		& \begin{minipage}{.18\textwidth}
			\center
			\includegraphics[width=\linewidth]{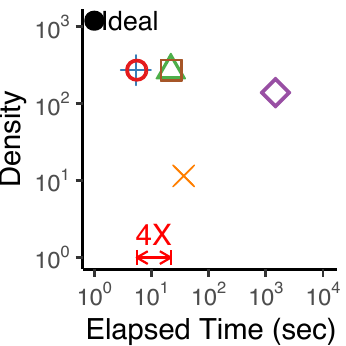}
		\end{minipage} 
		& \begin{minipage}{.18\textwidth}
			\center
			\includegraphics[width=\linewidth]{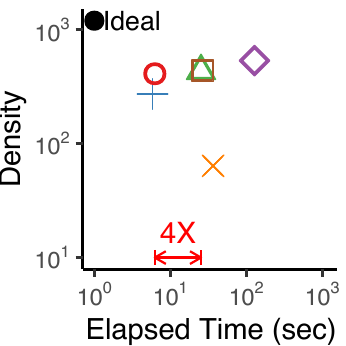}
		\end{minipage} 
		& \begin{minipage}{.18\textwidth}
			\center
			\includegraphics[width=\linewidth]{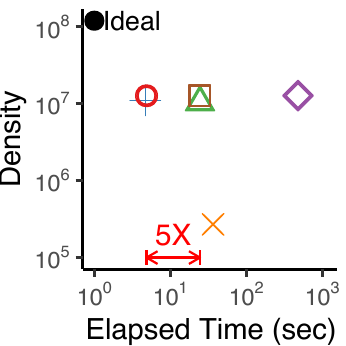}
		\end{minipage} 
		& \begin{minipage}{.18\textwidth}
			\center
			\includegraphics[width=\linewidth]{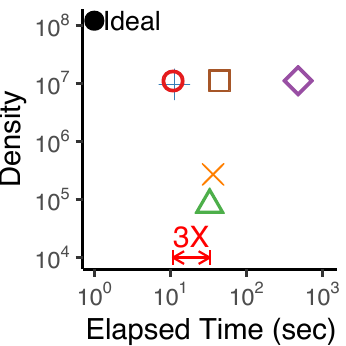}
		\end{minipage} \vspace{1mm} \\
		A. $\densitysuspnoarg$ & B. $\densityarinoarg$ & C. $\densitygeonoarg$ & D. $\densitysurpalpha{1}$ & E. $\densitysurpalpha{10}$
	\end{tabular}
	\caption{\label{fig:tradeoff:kowiki}
		Speed and accuracy of the algorithms in the KoWiki dataset.
	}
\end{figure}

\begin{figure}[h]
	\centering
	\includegraphics[width=\linewidth]{FIG/tradeoff_label.pdf} \\
	\begin{tabular}{ccccc}
		\begin{minipage}{0.18\textwidth}
			\center
			\includegraphics[width=\linewidth]{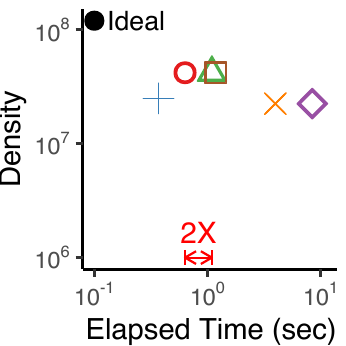}
		\end{minipage} 
		& \begin{minipage}{.18\textwidth}
			\center
			\includegraphics[width=\linewidth]{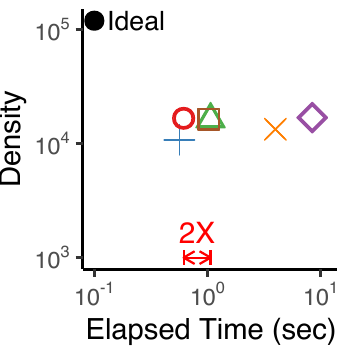}
		\end{minipage} 
		& \begin{minipage}{.18\textwidth}
			\center
			\includegraphics[width=\linewidth]{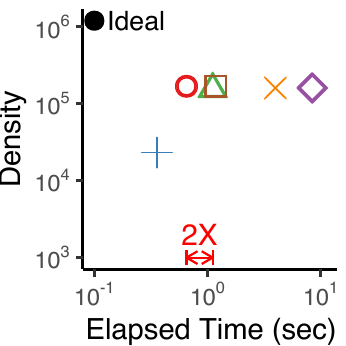}
		\end{minipage} 
		& \begin{minipage}{.18\textwidth}
			\center
			\includegraphics[width=\linewidth]{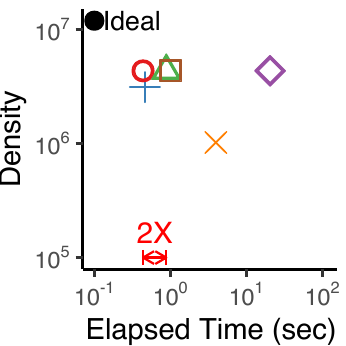}
		\end{minipage} 
		& \begin{minipage}{.18\textwidth}
			\center
			\includegraphics[width=\linewidth]{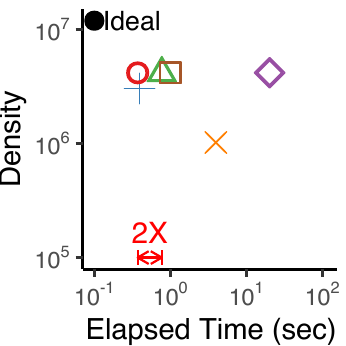}
		\end{minipage} \vspace{1mm} \\
		A. $\densitysuspnoarg$ & B. $\densityarinoarg$ & C. $\densitygeonoarg$ & D. $\densitysurpalpha{1}$ & E. $\densitysurpalpha{10}$
	\end{tabular}
	\caption{\label{fig:tradeoff:darpa}
		Speed and accuracy of the algorithms in the DARPA dataset.
	}
\end{figure}

\begin{figure}[h]
	\centering
	\includegraphics[width=\linewidth]{FIG/tradeoff_label.pdf} \\
	\begin{tabular}{ccccc}
		\begin{minipage}{0.18\textwidth}
			\center
			\includegraphics[width=\linewidth]{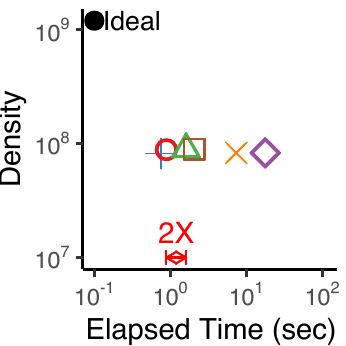}
		\end{minipage} 
		& \begin{minipage}{.18\textwidth}
			\center
			\includegraphics[width=\linewidth]{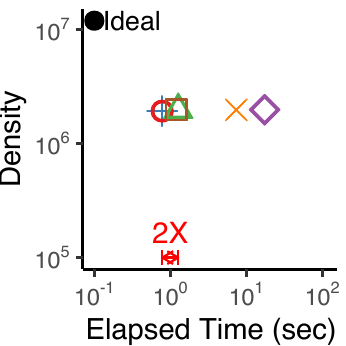}
		\end{minipage} 
		& \begin{minipage}{.18\textwidth}
			\center
			\includegraphics[width=\linewidth]{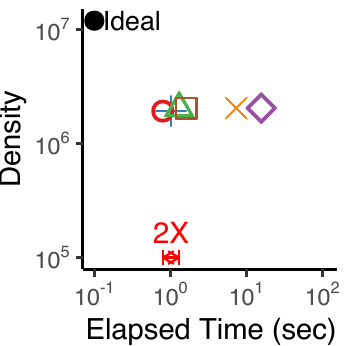}
		\end{minipage} 
		& \begin{minipage}{.18\textwidth}
			\center
			\includegraphics[width=\linewidth]{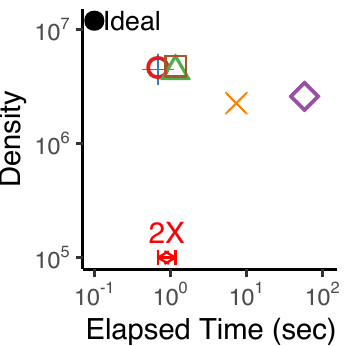}
		\end{minipage} 
		& \begin{minipage}{.18\textwidth}
			\center
			\includegraphics[width=\linewidth]{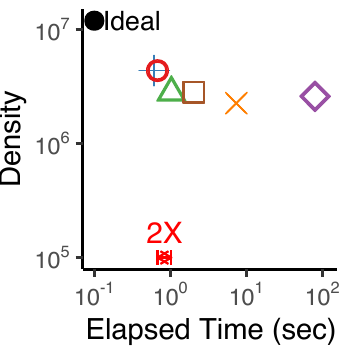}
		\end{minipage} \vspace{1mm} \\
		A. $\densitysuspnoarg$ & B. $\densityarinoarg$ & C. $\densitygeonoarg$ & D. $\densitysurpalpha{1}$ & E. $\densitysurpalpha{10}$
	\end{tabular}
	\caption{\label{fig:tradeoff:intrusion}
		Speed and accuracy of the algorithms in the AirForce dataset.
	}
\end{figure}

\begin{figure}[h]
	\centering
	\includegraphics[width=\linewidth]{FIG/tradeoff_label.pdf} \\
	\begin{tabular}{ccccc}
		\begin{minipage}{0.18\textwidth}
			\center
			\includegraphics[width=\linewidth]{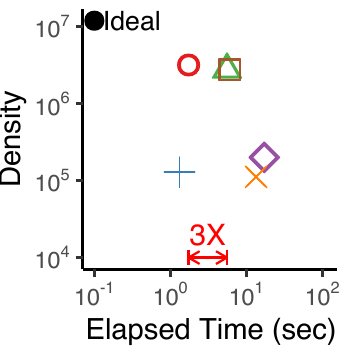}
		\end{minipage} 
		& \begin{minipage}{.18\textwidth}
			\center
			\includegraphics[width=\linewidth]{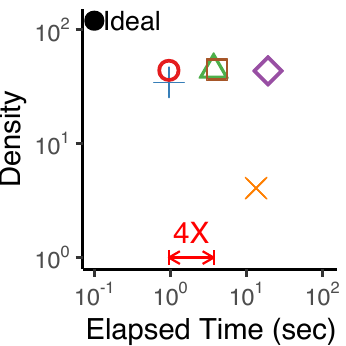}
		\end{minipage} 
		& \begin{minipage}{.18\textwidth}
			\center
			\includegraphics[width=\linewidth]{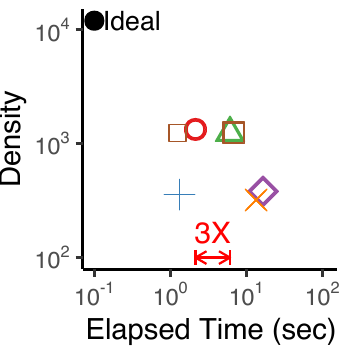}
		\end{minipage} 
		& \begin{minipage}{.18\textwidth}
			\center
			\includegraphics[width=\linewidth]{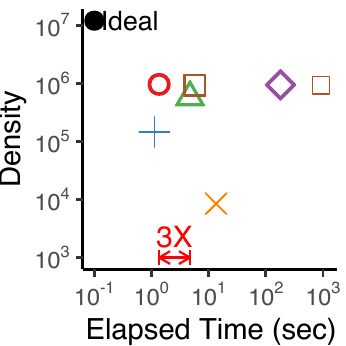}
		\end{minipage} 
		& \begin{minipage}{.18\textwidth}
			\center
			\includegraphics[width=\linewidth]{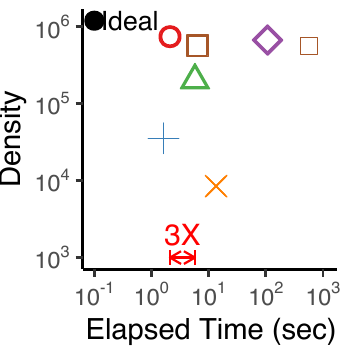}
		\end{minipage} \vspace{1mm} \\
		A. $\densitysuspnoarg$ & B. $\densityarinoarg$ & C. $\densitygeonoarg$ & D. $\densitysurpalpha{1}$ & E. $\densitysurpalpha{10}$
	\end{tabular}
	\caption{\label{fig:tradeoff:itunes}
		Speed and accuracy of the algorithms in the SWM dataset.
	}
\end{figure}

\begin{figure}[h]
	\centering
	\includegraphics[width=\linewidth]{FIG/tradeoff_label.pdf} \\
	\begin{tabular}{ccccc}
		\begin{minipage}{0.18\textwidth}
			\center
			\includegraphics[width=\linewidth]{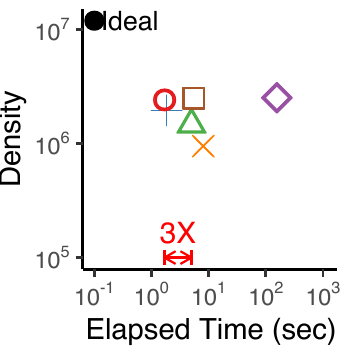}
		\end{minipage} 
		& \begin{minipage}{.18\textwidth}
			\center
			\includegraphics[width=\linewidth]{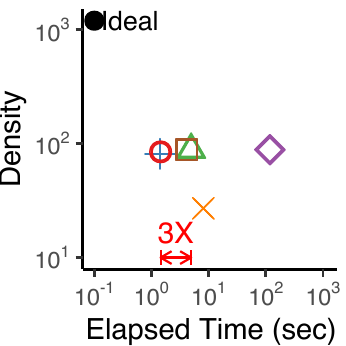}
		\end{minipage} 
		& \begin{minipage}{.18\textwidth}
			\center
			\includegraphics[width=\linewidth]{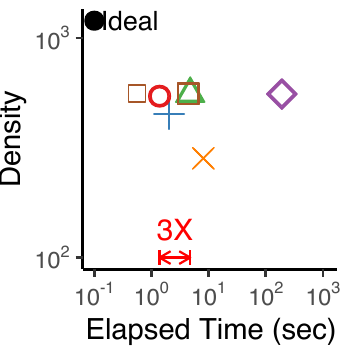}
		\end{minipage} 
		& \begin{minipage}{.18\textwidth}
			\center
			\includegraphics[width=\linewidth]{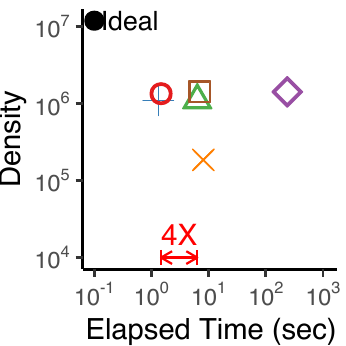}
		\end{minipage} 
		& \begin{minipage}{.18\textwidth}
			\center
			\includegraphics[width=\linewidth]{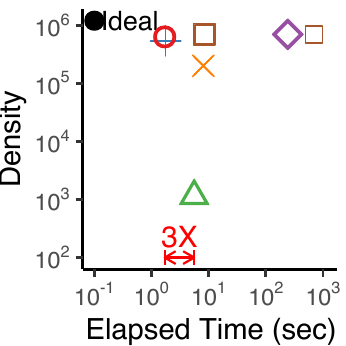}
		\end{minipage} \vspace{1mm} \\
		A. $\densitysuspnoarg$ & B. $\densityarinoarg$ & C. $\densitygeonoarg$ & D. $\densitysurpalpha{1}$ & E. $\densitysurpalpha{10}$
	\end{tabular}
	\caption{\label{fig:tradeoff:yelp}
		Speed and accuracy of the algorithms in the Yelp dataset.
	}
\end{figure}

\begin{figure}[h]
	\centering
	\includegraphics[width=\linewidth]{FIG/tradeoff_label.pdf} \\
	\begin{tabular}{ccccc}
		\begin{minipage}{0.18\textwidth}
			\center
			\includegraphics[width=\linewidth]{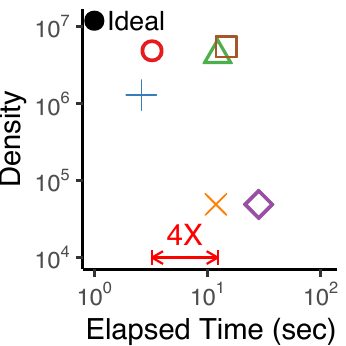}
		\end{minipage} 
		& \begin{minipage}{.18\textwidth}
			\center
			\includegraphics[width=\linewidth]{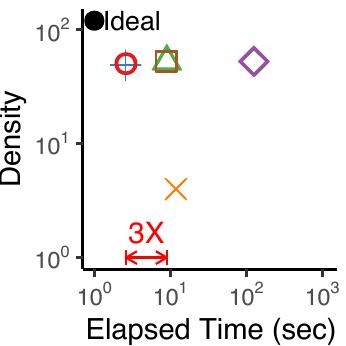}
		\end{minipage} 
		& \begin{minipage}{.18\textwidth}
			\center
			\includegraphics[width=\linewidth]{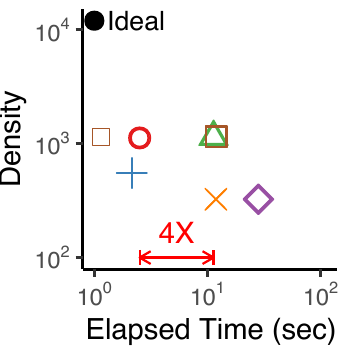}
		\end{minipage} 
		& \begin{minipage}{.18\textwidth}
			\center
			\includegraphics[width=\linewidth]{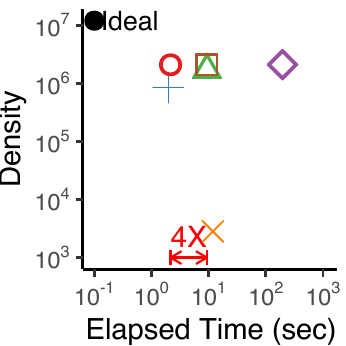}
		\end{minipage} 
		& \begin{minipage}{.18\textwidth}
			\center
			\includegraphics[width=\linewidth]{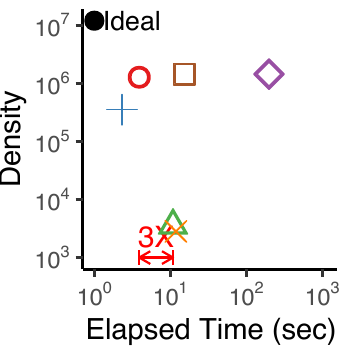}
		\end{minipage} \vspace{1mm} \\
		A. $\densitysuspnoarg$ & B. $\densityarinoarg$ & C. $\densitygeonoarg$ & D. $\densitysurpalpha{1}$ & E. $\densitysurpalpha{10}$
	\end{tabular}
	\caption{\label{fig:tradeoff:android}
		Speed and accuracy of the algorithms in the Android dataset.
	}
\end{figure}

\begin{figure}[h]
	\centering
	\includegraphics[width=\linewidth]{FIG/tradeoff_label.pdf} \\
	\begin{tabular}{ccccc}
		\begin{minipage}{0.18\textwidth}
			\center
			\includegraphics[width=\linewidth]{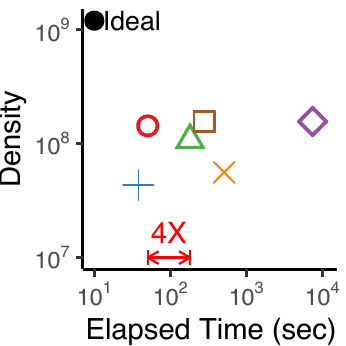}
		\end{minipage} 
		& \begin{minipage}{.18\textwidth}
			\center
			\includegraphics[width=\linewidth]{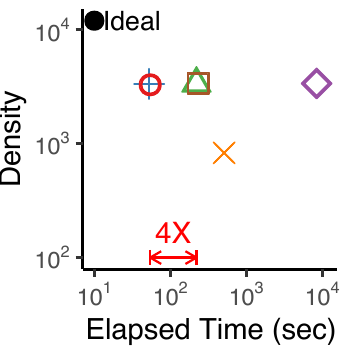}
		\end{minipage} 
		& \begin{minipage}{.18\textwidth}
			\center
			\includegraphics[width=\linewidth]{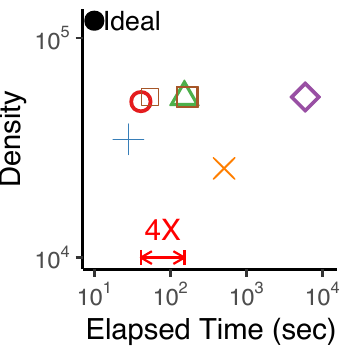}
		\end{minipage} 
		& \begin{minipage}{.18\textwidth}
			\center
			\includegraphics[width=\linewidth]{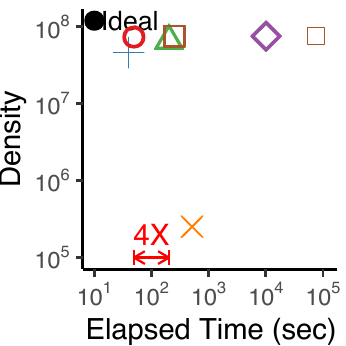}
		\end{minipage} 
		& \begin{minipage}{.18\textwidth}
			\center
			\includegraphics[width=\linewidth]{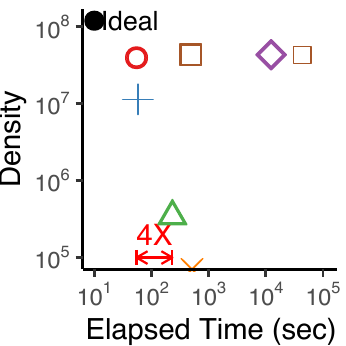}
		\end{minipage} \vspace{1mm} \\
		A. $\densitysuspnoarg$ & B. $\densityarinoarg$ & C. $\densitygeonoarg$ & D. $\densitysurpalpha{1}$ & E. $\densitysurpalpha{10}$
	\end{tabular}
	\caption{\label{fig:tradeoff:netflix}
		Speed and accuracy of the algorithms in the Netflix dataset.
	}
\end{figure}

\begin{figure}[h]
	\centering
	\includegraphics[width=\linewidth]{FIG/tradeoff_label.pdf} \\
	\begin{tabular}{ccccc}
		\begin{minipage}{0.18\textwidth}
			\center
			\includegraphics[width=\linewidth]{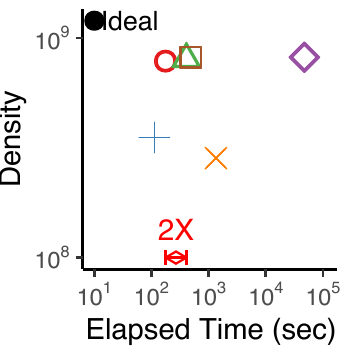}
		\end{minipage} 
		& \begin{minipage}{.18\textwidth}
			\center
			\includegraphics[width=\linewidth]{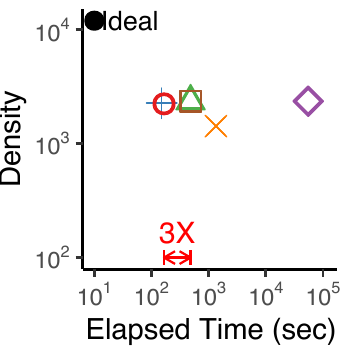}
		\end{minipage} 
		& \begin{minipage}{.18\textwidth}
			\center
			\includegraphics[width=\linewidth]{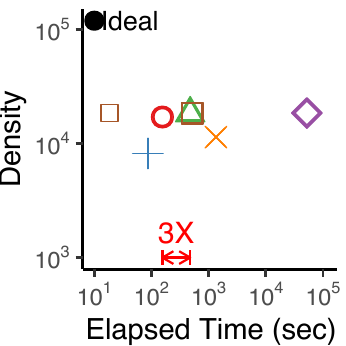}
		\end{minipage} 
		& \begin{minipage}{.18\textwidth}
			\center
			\includegraphics[width=\linewidth]{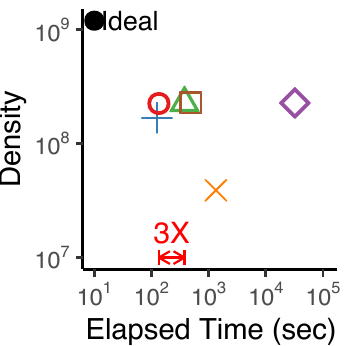}
		\end{minipage} 
		& \begin{minipage}{.18\textwidth}
			\center
			\includegraphics[width=\linewidth]{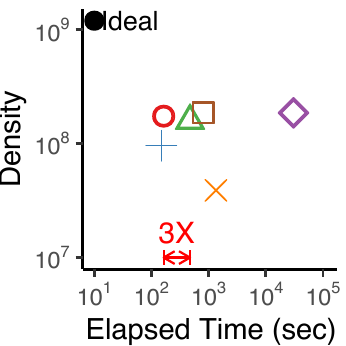}
		\end{minipage} \vspace{1mm} \\
		A. $\densitysuspnoarg$ & B. $\densityarinoarg$ & C. $\densitygeonoarg$ & D. $\densitysurpalpha{1}$ & E. $\densitysurpalpha{10}$
	\end{tabular}
	\caption{\label{fig:tradeoff:yahoo}
		Speed and accuracy of the algorithms in the YahooM. dataset.
	}
\end{figure}

\end{document}